\documentclass[11pt,a4paper]{article}
\usepackage[T1]{fontenc}
\usepackage[utf8]{inputenc}
\usepackage[sc]{mathpazo}
\usepackage[a4paper,left=24mm,right=24mm,top=24mm,bottom=25mm,headheight=15pt,headsep=9mm,footskip=12mm]{geometry}
\usepackage{amsmath,amssymb,amsthm,mathtools,bm,mathrsfs}
\usepackage{graphicx,booktabs,microtype,enumitem,array,longtable}
\usepackage[section]{placeins}
\usepackage{flafter}
\usepackage{titlesec,fancyhdr,xcolor}
\usepackage[round,authoryear]{natbib}
\IfFileExists{xurl.sty}{\usepackage{xurl}}{}
\definecolor{inkblue}{RGB}{26,55,75}
\definecolor{quietgray}{RGB}{90,97,103}
\usepackage[colorlinks=true,linkcolor=inkblue,citecolor=inkblue,urlcolor=inkblue,bookmarksnumbered=true]{hyperref}
\hypersetup{pdftitle={When Can World Models Recover Physical Laws?},
  pdfauthor={Ye Yuan and Jun Liu},
  pdfsubject={Technical report: information, identifiability, and experimental cost},
  pdfkeywords={world models, physical-law recovery, information theory, experimental design}}
\newtheorem{theorem}{Theorem}[section]
\newtheorem{proposition}[theorem]{Proposition}
\newtheorem{corollary}[theorem]{Corollary}
\theoremstyle{definition}
\newtheorem{assumption}[theorem]{Assumption}
\theoremstyle{plain}
\newcommand{\chg}[1]{{\color{blue}#1}}
\newcommand{\W}{\mathcal W}
\newcommand{\A}{\mathcal A}
\newcommand{\F}{\mathcal F}
\newcommand{\E}{\mathbb E}
\newcommand{\R}{\mathbb R}
\newcommand{\KL}{\operatorname{KL}}
\newcommand{\TV}{\operatorname{TV}}
\newcommand{\diam}{\operatorname{diam}}

\newcommand{\argmin}{\operatorname*{arg\,min}}
\newcommand{\lawhat}{\widehat F}
\newcommand{\decoder}{\mathsf D}
\newcommand{\loss}{\rho}
\newcommand{\cost}{\mathsf c}
\newcommand{\capinfo}{\mathcal C}
\newcommand{\rd}{\mathcal R}

\setlist{nosep,leftmargin=1.5em}
\titleformat{\section}{\Large\bfseries\color{inkblue}}{\thesection}{0.7em}{}
\titlespacing*{\section}{0pt}{2.0ex plus 0.8ex minus 0.2ex}{0.9ex}
\titleformat{\subsection}{\large\bfseries}{\thesubsection}{0.7em}{}
\fancypagestyle{plain}{\fancyhf{}\fancyfoot[C]{\small\thepage}}
\begin{document}
\raggedbottom
\thispagestyle{plain}
\vspace*{8mm}
\noindent{\sffamily\small\color{quietgray}\MakeUppercase{Technical Report}\hfill September 2026}\par
\vspace{5mm}
\noindent{\color{inkblue}\rule{\linewidth}{0.8pt}}\par
\vspace{8mm}
\noindent{\fontsize{29}{34}\selectfont\bfseries\color{inkblue}
When Can World Models\\Recover Physical Laws?\par}
\vspace{7mm}
\noindent Ye Yuan\textsuperscript{1}\quad and\quad Jun Liu\textsuperscript{2}\par
\vspace{3mm}
{\small
\noindent\textsuperscript{1}School of Artificial Intelligence and Automation,\\
Huazhong University of Science and Technology, Wuhan, China\\
\href{mailto:yye@hust.edu.cn}{\texttt{yye@hust.edu.cn}}\par
\vspace{2mm}
\noindent\textsuperscript{2}Department of Applied Mathematics,\\
University of Waterloo, Waterloo, Canada\\
\href{mailto:j.liu@uwaterloo.ca}{\texttt{j.liu@uwaterloo.ca}}\par
}
\vspace{2mm}
\noindent{\small\color{quietgray}15 September 2026}\par
\vspace{12mm}
\begin{abstract}
Accurate prediction does not establish that a world model has recovered a physical law: distinct dynamics can generate identical records under the same observation protocol. We formulate law recovery on a fixed physical domain under an explicit catalog of experiments, sensor uncertainty, and an acquisition budget. A rate--distortion converse separates the information needed to describe a law from the information the apparatus can reveal. Its constructive counterpart gives a finite response codebook and an explicit decoding budget. On compact world classes, uniform recovery is possible exactly when every pair of different laws is experimentally distinguishable; equivalently, the apparatus can recover all the entropy of every finite law source. An inverse response modulus quantifies stability. For Lipschitz fields on a $d$-dimensional state--action domain, noisy full-state readouts after resets require minimax budget $\Theta(\varepsilon^{-(d+4)/2})$ for squared law error $\varepsilon$, compared with $\Theta(\varepsilon^{-(d+2)/2})$ for direct field observations. Exact crossing-time symmetries establish the lower bound under adaptive experiment selection and arbitrary durations with constant inputs. Reproducible synthetic cases illustrate the separate roles of intervention, calibration, and repeated measurement. Together, the results identify which evidence supports a claim of physical-law recovery and the cost of acquiring it.
\end{abstract}

\vfill
\clearpage
\tableofcontents
\vspace{9mm}

\clearpage
\section{Introduction}\label{sec:intro}
An agent that acts in a physical environment must anticipate the consequences of its actions. Before moving an object, changing a machine setting, or applying a force, it benefits from knowing what is likely to happen next. A world model provides such predictions from previous observations and actions. The scientific question considered in this report begins one step further: when do the observations support recovery of the physical law that generates those consequences? We ask what an experiment must reveal, how accurately the law can then be recovered, and what acquiring that evidence costs.

The distinction can be understood without choosing a learning algorithm. Suppose a model predicts the next position of a moving object accurately throughout a recording. This establishes agreement with the observed motion. It does not yet establish what the object would do from another initial position, at another velocity, or under an intervention absent from the recording. A physical law makes these additional predictions through a common rule of evolution. Recovering that rule therefore requires deciding between explanations that agree on the available observations but differ on the physical situations of interest. More examples of the same motion need not resolve the decision.

Consider an object moving around a fixed center, as in Figure~\ref{fig:orbit}. One candidate law is inverse-square attraction: the acceleration magnitude decreases as the square of the distance increases. Another is a linear restoring law: the acceleration magnitude increases in proportion to distance. Choose units in which both laws give unit inward acceleration at unit radius. The accelerations at planar position $\boldsymbol q$ are then $a_G(\boldsymbol q)=-\boldsymbol q/\|\boldsymbol q\|^3$ and $a_H(\boldsymbol q)=-\boldsymbol q$. Start the object at unit radius with unit tangential speed. Both laws produce exactly the same circular motion, with the same position at every time. Even an arbitrarily long, noiseless recording cannot identify which law generated it.

\begin{figure}[!htbp]
\centering
\includegraphics[width=\linewidth]{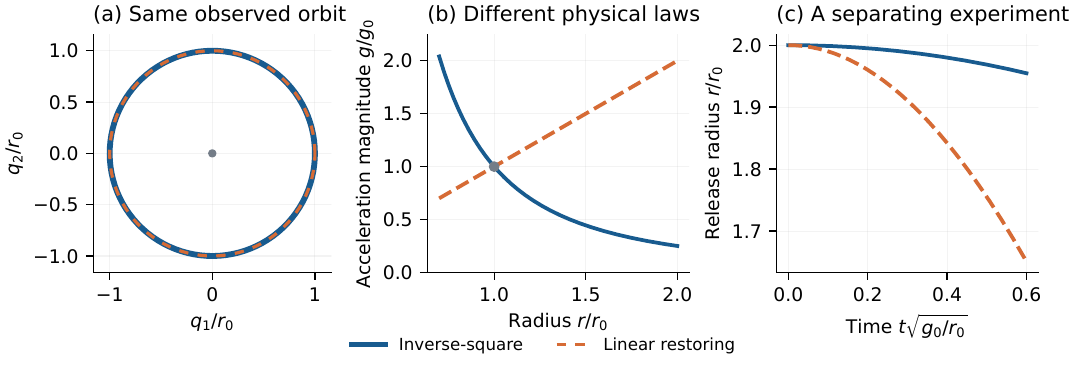}
\caption{Why accurate prediction of a recorded motion need not identify its physical law. Blue solid curves represent inverse-square attraction and orange dashed curves a linear restoring force. The laws are normalized to have the same inward acceleration $g_0$ at radius $r_0$; position, acceleration, and time are measured in units of $r_0$, $g_0$, and $\sqrt{r_0/g_0}$. (a) The axes show the two planar position coordinates. Starting at radius $r_0$ with tangential speed $\sqrt{r_0g_0}$ gives the same circular trajectory under both laws, including the same timing. Their superposed curves are therefore observationally indistinguishable even with exact, continuous position measurements. (b) Acceleration magnitude versus radius reveals the hidden disagreement: $g/g_0=(r/r_0)^{-2}$ for inverse-square attraction and $g/g_0=r/r_0$ for the restoring law. They meet at the observed radius but differ elsewhere. (c) A different preparation exposes that disagreement. Both trials now start at radius $2r_0$ with zero velocity, and the vertical axis records radius as a function of time. The restoring law pulls the object inward faster because its initial acceleration is $2g_0$, compared with $g_0/4$ for inverse-square attraction. These are noiseless trajectories. Section~\ref{sec:cases} adds Gaussian readout noise and quantifies the benefit of repeated releases.}
\label{fig:orbit}
\end{figure}

The disagreement becomes visible when the experiment changes. Release the object from rest at radius two. Its initial inward acceleration is now $1/4$ under inverse-square attraction and $2$ under the restoring law, so the position responses separate. This is a change in preparation, not an external force applied during the motion. Repeating the circular record only refines an already shared prediction; repeating the release can distinguish the two explanations despite measurement noise. The example will recur throughout the report, first as a recovery problem, then as an information constraint, and finally as a concrete experiment with a calculable error probability.

For this example, the physical state is position and velocity together, $x=(\boldsymbol q,\boldsymbol v)$, and the corresponding first-order laws are $F_G(x)=(\boldsymbol v,a_G(\boldsymbol q))$ and $F_H(x)=(\boldsymbol v,a_H(\boldsymbol q))$. More generally, we study a field $F$ in an ordinary differential equation $\dot x=F(x,u)$, where $u$ denotes a controlled input when one is available. Recovery means estimating this field over a declared physical domain. It does not require discovering a unique symbolic expression: different expressions can define the same field. It does require specifying the coordinates, region, and accuracy in which two candidate fields are to be compared.

This question complements the predictive use of world models. \citet{ha2018} learned compact representations of environments for control, and PlaNet planned through learned latent dynamics from image observations \citep{hafner2019planet}. Dreamer subsequently demonstrated the usefulness of world models across diverse control tasks \citep{hafner2025}. These results make prediction and planning compelling objectives in their own right. Our concern is the additional evidence needed when a model is interpreted as recovering a physical rule. A representation useful for choosing actions need not uniquely determine the continuous-time field in specified physical coordinates.

Related distinctions arise in representation learning. I-JEPA predicts representations of unobserved image regions from visible context, while V-JEPA learns from predicting features of video \citep{assran2023ijepa,bardes2024vjepa}. Such objectives can learn useful abstractions without reconstructing every observed detail. They do not, by themselves, specify a physical-coordinate recovery target or an experiment that identifies it. In another setting, \citet{locatello2019} showed that unsupervised disentanglement requires assumptions on models and data. Causal representation learning likewise asks how learned variables can support a causal description \citep{schoelkopf2021}. These questions concern latent representations rather than the dynamical fields studied here, but they reinforce the need to specify what fixes the physical meaning of an estimate. Our coordinates and sensing assumptions are explicit parts of the recovery problem.

The issue has a long history in system identification. Structural identifiability asks whether different systems can generate the same observations \citep{bellman1970}; symmetries of dynamics and sensing explain why such equivalences can persist \citep{sontag2017}. In the circular example, the available motion never visits radii where the laws differ. With an uncertain sensor, an additional ambiguity is possible: changing the physical scale while compensating through sensor gain can leave all readings unchanged. The report treats both as properties of the permitted experiments. An informative release addresses the first obstruction; a known physical reference may be necessary for the second.

Recovering equations from data also builds on substantial constructive work. The sparse Bayesian approach of \citet{pan2016}  identifies nonlinear state-space systems using sparse structure; SINDy applies similar idea by selecting governing terms from a candidate library \citep{brunton2016}. Neural ordinary differential equations instead parameterize continuous-time dynamics through trainable vector fields \citep{chen2018}. Active identification provides guarantees for nonlinear dynamics linear in known features \citep{mania2022}, and continuous-time sampling theory characterizes when sampled evolution determines its generator \citep{zeng2024}. These methods and results answer different parts of the recovery problem. Our formulation makes the target domain, sensor uncertainty, executable interventions, and acquisition costs explicit within a common analysis.

The starting point is thus the experimental record, before selecting a representation or optimizer. An apparatus supplies a catalog of operations: preparing states, applying specified inputs, observing at chosen times, and, when available, calibrating against a reference. A learner may choose the next operation using earlier readings. Every operation has a cost, and the complete record must fit an acquisition budget. The central difficulty is whether this record separates candidate laws to the required physical precision. This perspective permits limits that apply to every learning algorithm, together with constructions that show when those limits can be approached.

Information theory provides a useful language for separating the demands of the target from the abilities of the apparatus \citep{shannon1948,shannon1959,cover2006}. A law class may contain many physically different possibilities at a given resolution; describing the correct one requires information. Experiments deliver information through their responses. In the equal-prior two-law example, identifying the correct candidate requires one bit, but the circular experiment reveals none of it. The release can reveal that distinction. With a continuous law class, metric entropy and rate--distortion quantify the corresponding description requirement, while experimental information capacity describes what can be learned within the budget.

There is an essential difference from ordinary communication. An experimenter does not know the true law and cannot encode it freely into a channel input. The experimenter can only probe the unknown world using available operations. Consequently, an inequality comparing required and available information is a necessary condition; it does not alone construct a successful recovery procedure. We give a separate response-matching decoder based on a finite collection of representative worlds. Under the stated compactness and continuity assumptions, arbitrary uniform accuracy is possible exactly when every pair of different target laws can be separated by some permitted experiment. The separating program may depend on the requested precision.

Accuracy also depends on how strongly responses separate. Closely spaced responses require more repetitions in noise, and a small response error can correspond to a large physical-law error. Our analysis uses established tools from minimax information bounds, sequential experimental design, and inverse-problem stability \citep{yang1999,yu1997,chernoff1959,kaufmann2016,donoho1991,donoho1994}. For a concrete class of Lipschitz fields, we then quantify how preparation, readout, and elapsed time determine acquisition cost. The squared-error exponent $-2/(d+4)$ is already established for nonparametric ODE estimation under specified observation designs \citep{schoetz2026,schoetz_siebel2026}; our costed access model examines it with adaptive experiment selection and contrasts state observations with direct field observations.

The argument proceeds in this order. Section~\ref{sec:problem} defines the physical target and available experiments. Section~\ref{sec:distinguishability} asks which laws those experiments can separate. Section~\ref{sec:information} compares description requirements with accessible information, and Section~\ref{sec:recovery} constructs a decoder and explains its stability. Section~\ref{sec:budget} specializes the theory to acquisition costs for Lipschitz dynamics. Section~\ref{sec:cases} returns to the orbital example and adds calibration and noisy-estimation cases. Each stage addresses a question left open by the previous one: what is being recovered, whether the evidence determines it, how accurately it can be recovered, and how much that accuracy costs.

\section{Problem Formulation}\label{sec:problem}
We ask how accurately a learner can recover a physical law from experiments whose total cost is limited. The formulation specifies the unknown system, the field to be recovered, the observations available to the learner, and the resource constraint, in that order.

For integers $p,s\ge1$ and $q\ge0$, let $x(t)\in\mathcal X\subseteq\R^p$ be a state in fixed physical coordinates and $u(t)\in U\subseteq\R^q$ a controlled input. A possible world consists of a dynamical law and a sensor:
\begin{equation}\label{eq:world}
 w=(F_w,H_w)\in\W,\qquad \dot x(t)=F_w(x(t),u(t)).
\end{equation}
The field $F_w:\mathcal X\times U\to\R^p$ determines the state velocity. The map $H_w:\mathcal X\to\R^s$ determines the noiseless sensor output, which has $s$ coordinates. The candidate class $\W$ is known, but the true member $w$ is not. Sensor uncertainty is included because different dynamics and sensors can explain the same observations. The target of recovery is the physical field $F_w$; the sensor need not itself be returned by the learner.

We evaluate an estimated field $\lawhat$ on a fixed compact domain $D=K\times U$, where $K\subseteq\mathcal X$ is a compact state region of interest and $U$ is compact. The system may move outside $K$ during an experiment; its dynamics must remain defined on the larger state region $\mathcal X$. Choose a probability measure $\nu$ on $D$ before collecting data. It assigns the evaluation weights, with positive mass in every nonempty relatively open part of $D$. Define
\begin{equation}\label{eq:metric}
 \loss(F,\widehat F)^2
 =\int_D\|F(x,u)-\widehat F(x,u)\|_2^2\,d\nu(x,u).
\end{equation}
For uniform $\nu$, this is the mean squared discrepancy between the true and estimated state velocities across the domain. It evaluates the function on $D$, including points absent from the training trajectories. We allow any square-integrable estimated field on $D$ and denote this output space by $\F$. A world model can use any internal representation, but a physical-law recovery claim requires such a physical-coordinate output; predicting sensor readings alone does not supply it.

To instantiate these objects, keep the two laws of Figure~\ref{fig:orbit} as a running example. The state is planar position and velocity, $x=(\boldsymbol q,\boldsymbol v)$, so $p=4$. There is no continuously applied input: the input dimension is $q=0$, and $U$ is a singleton. An intervention will instead change the initial preparation. On $\mathcal X=(\R^2\setminus\{0\})\times\R^2$, the two candidate fields are $F_G(x)=(\boldsymbol v,-\boldsymbol q/\|\boldsymbol q\|^3)$ and $F_H(x)=(\boldsymbol v,-\boldsymbol q)$. Take a known radial sensor $H(x)=\|\boldsymbol q\|$, with $s=1$. The corresponding worlds are $w_G=(F_G,H)$ and $w_H=(F_H,H)$. The picture displays planar geometry, but the measurement used in the example is one scalar radius.

Choose $K$ to consist of all states with position radii in $[1,2]$ and velocities in the closed unit ball. Specify $\nu$ by drawing radius uniformly on $[1,2]$, angle uniformly on $[0,2\pi)$, and velocity uniformly in that ball, independently. This is uniform radial weighting, not uniform planar-area weighting. The fields have the same position derivative and different accelerations. Their distance is therefore
\begin{equation}\label{eq:orbit-law-distance}
 \delta_{\rm law}^{\,2}
 :=\loss(F_G,F_H)^2
 =\int_1^2(r-r^{-2})^2\,dr
 =\frac{21}{8}-2\log 2.
\end{equation}
This positive number measures disagreement over the evaluation region, even though the laws coincide on the unit circle. The chosen region and weighting are part of the recovery question: evaluating only that circle would remove the distinction that this example is intended to study.

An apparatus provides a catalog $\A$ of executable experiments, called episodes. An episode $a\in\A$ specifies its preparation, its input signal, and a finite readout schedule before execution. For a trajectory episode, let $x_w^a$ be its state trajectory and $t_{a,1},\ldots,t_{a,J_a}$ its $J_a\ge1$ readout times. The learner observes a vector $Y_a\in\R^{k_a}$, with $k_a=sJ_a$, given by
\begin{equation}\label{eq:episode}
 Y_a=m_a(w)+\xi_a,\qquad
 m_a(w)=\bigl(H_w(x_w^a(t_{a,i}))\bigr)_{i=1}^{J_a},\qquad
 \xi_a\sim N(0,\Sigma_a).
\end{equation}
Thus $m_a(w)$ is the predicted observation of an experiment, not the field itself. The covariance $\Sigma_a\succ0$ is known and does not depend on the world. Noise is independent across executions; readouts within one episode may be correlated. A calibration episode instead measures $H_w(x_{\rm ref})$ at a known physical reference $x_{\rm ref}\in\mathcal X$, and is available only when the apparatus supplies that operation.

In the same example, a passive episode starts at $\boldsymbol q(0)=(1,0)$, $\boldsymbol v(0)=(0,1)$. Both laws then give $H(x(t))=1$ at every readout. A release episode starts instead at $\boldsymbol q(0)=(2,0)$, $\boldsymbol v(0)=0$ and reads the radius at $t_*=0.4$. Its two means differ by approximately $0.13781$; taking $\sigma=0.1$ gives $\Sigma_a=\sigma^2$ for this single readout. Changing the preparation has changed the observation map $m_a$ without changing the sensor or the quantity to be recovered. Only trajectories that remain away from the singular origin through their readout are admitted; they need not stay inside the evaluation set $K$.

A sampling policy $\pi$ uses completed observations to choose the next episode or to stop. It chooses experiments; the input $u(t)$ is one component of the chosen experiment. Within an episode its controls and readout schedule remain fixed. Let $N_\pi$ be the number of executed episodes and let $T^\pi$ be the complete record of their specifications and readings, $(a_1,Y_1,\ldots,a_{N_\pi},Y_{N_\pi})$, together with any learner-accessible policy randomness. A learning rule $\decoder$ maps this record to the field estimate $\lawhat=\decoder(T^\pi)\in\F$. Policy and decoder randomness are independent of the world and measurement noise; formal measurability conditions are in Appendix~\ref{app:model}.

Each episode costs $c_a\ge c_{\min}>0$, accounting for the preparations, readouts, and elapsed time charged by the apparatus. A budget $B\ge0$ limits their total acquisition cost. Define $\Pi_{\A}(B)$ to be the set of admissible policies satisfying $\sum_{j=1}^{N_\pi}c_{a_j}\le B$ on every realized run. Consequently $N_\pi\le\lfloor B/c_{\min}\rfloor$. This budget measures data acquisition; model size and training computation are not charged.

For example, charging one unit per independent release allows at most $\lfloor B\rfloor$ releases. Additional optimization on the same record may improve a particular estimator, but it cannot purchase another physical observation. This separation lets us ask what any architecture could achieve before studying its computational implementation.

The best uniform recovery error at budget $B$ is
\begin{equation}\label{eq:risk}
 \mathfrak R_{\A}^\star(B)
 =\inf_{\pi\in\Pi_{\A}(B),\,\decoder}\sup_{w\in\W}
       \E_w\loss\bigl(F_w,\decoder(T^\pi)\bigr)^2.
\end{equation}
For a fixed world, $\E_w$ averages measurement noise and algorithm randomness. The supremum takes the worst candidate world; the infimum chooses the best single experiment-and-decoder procedure for the whole class. The reverse question is the least budget that permits a prescribed accuracy:
\begin{equation}\label{eq:budgetfrontier}
 B_{\A}^\star(\varepsilon)
 =\inf\!\left\{B\chg{\ge0}:\exists\pi\in\Pi_{\A}(B),\decoder\text{ with }
       \sup_{w\in\W}\E_w\loss\bigl(F_w,\decoder(T^\pi)\bigr)^2\le\varepsilon\right\}.
\end{equation}
Here $\varepsilon$ is expected squared field error, $\sqrt\varepsilon$ is its root-mean-square scale, and $\inf\varnothing=+\infty$. Both definitions range over all representations and learning rules, rather than a specified neural architecture.

The results of Sections~\ref{sec:distinguishability}--\ref{sec:recovery} use the following standing assumption.
\begin{assumption}[Compact world class]\label{ass:standing}
\begin{enumerate}[label=(\roman*)]
 \item The world class $\W$ is a nonempty compact metric space with its Borel $\sigma$-algebra.
 \item For every $w\in\W$, the target field $F_w|_D$ is continuous on $D$, and the map $w\mapsto F_w|_D$ is continuous from $\W$ into $(\F,\loss)$.
 \item For every episode $a\in\A$, the covariance $\Sigma_a\succ0$ does not depend on $w$, the cost satisfies $c_{\min}\le c_a<\infty$, and the observation mean $w\mapsto m_a(w)$ is continuous on $\W$.
 \item Every trajectory used by an episode exists, is unique, and remains in $\mathcal X$ through its final readout.
 \item At least two worlds have different target fields, so that the diameter $D_F=\max_{w,v\in\W}\loss(F_w,F_v)$, which is attained by (i)--(ii), is strictly positive.
\end{enumerate}
\end{assumption}
Compactness and continuity allow finite approximations at any fixed precision; they do not require finitely many physical parameters. Appendix~\ref{app:model} gives the measurability conventions for policies, transcripts, and decoders.

\section{What Experiments Can Distinguish}\label{sec:distinguishability}
Before asking how many measurements are needed, consider what could be learned if an experiment's mean response were known exactly. The circular-orbit example shows why this question comes first: perfect knowledge of that response still leaves two different laws possible.

The simplest obstruction is exact agreement of experimental responses. If two worlds give the same distribution of readings for every permitted episode, no choice of later experiments can distinguish them. In the present model the noise covariance is known and world-independent, so equality of these distributions is exactly equality of their means. Define
\begin{equation}\label{eq:equivalence}
 w\sim_{\A}v\ \Longleftrightarrow\ m_a(w)=m_a(v)\ \forall a\in\A,
 \qquad
 \kappa_{\A}=\max_{w\sim_{\A}v}\loss(F_w,F_v).
\end{equation}
Equivalent worlds induce identical complete transcript laws, including under adaptive episode selection. The triangle inequality then yields the ambiguity floor.
\begin{proposition}[Ambiguity floor]\label{prop:floor}
Under Assumption~\ref{ass:standing}, the maximum defining $\kappa_{\A}$ in~\eqref{eq:equivalence} is attained. If $w\sim_{\A}v$, then every policy $\pi$, including an adaptive policy with randomized stopping, induces the same transcript law in $w$ and in $v$. Consequently
\begin{equation}\label{eq:floor}
 \mathfrak R_{\A}^\star(B)\ge\kappa_{\A}^2/4\qquad\text{for every }B\ge0.
\end{equation}
\end{proposition}
Appendix~\ref{app:information} proves Proposition~\ref{prop:floor}.
Repeating the same kind of evidence cannot recover the missing distinction; a larger predictive model does not alter this floor.

For the running example, first restrict the apparatus to repetitions of the circular preparation, with any finite readout schedule. The two noiseless response vectors consist entirely of ones. Independent noise gives more precise estimates of those ones, but no evidence favoring either force law. On the two-world class $\W=\{w_G,w_H\}$, equation~\eqref{eq:equivalence} gives $\kappa_{\A}=\delta_{\rm law}$. Consequently every procedure has worst-case squared field error at least $\delta_{\rm law}^2/4$, with the distance evaluated in equation~\eqref{eq:orbit-law-distance}.

Here the bound can be attained without data: return the midpoint field $(F_G+F_H)/2$. Its squared distance from either candidate is $\delta_{\rm law}^2/4$. This is allowed because the output space contains arbitrary square-integrable fields, including estimates outside the candidate class. A different decision problem would require selecting one of the two named laws. Under equal priors that selection has error probability $1/2$, and random selection gives squared field risk $\delta_{\rm law}^2/2$. A testing error and a field-recovery error therefore have different meanings and numerical scales, even on this simple example.

Now permit the release preparation. Its unequal radial means make the two worlds inequivalent, so $\kappa_{\A}=0$ for this two-world class. A single noisy reading still does not identify the law perfectly: removal of an exact ambiguity only makes repeated informative measurement potentially useful. Conversely, enlarging the evaluation region can expose additional disagreements between fields while leaving the permitted observations unchanged. Both what must be recovered and what can be observed must therefore be specified. The release establishes separability of these two candidates; it does not show that one release can recover an unrestricted unknown force field. A richer candidate class may contain further laws that agree on both preparations and require additional experiments.

The quantity $\kappa_{\A}$ expresses an obstruction of the apparatus, rather than an error caused by a finite sample. It depends on the competing laws, their possible sensors, and the permitted experiments together. If it is positive, some physical distinction remains inaccessible even as the budget grows. If it is zero, every pair of different laws can be separated by at least one experiment. Section~\ref{sec:recovery} will show why compactness turns this pairwise condition into uniform recovery at every positive tolerance.

We now turn from distinctions that are absent to distinctions that must be measured accurately. This requires a measure of the information needed to describe the law and a separate measure of the information the apparatus can reveal.

\section{Information Required and Information Available}\label{sec:information}
Shannon's theory separates the information required to describe a source from the information a channel can convey \citep{shannon1948,shannon1959}. For physical-law learning, these are different objects: complexity belongs to the target field on $D$, while accessible information belongs to the experiments and their sensors. We first consider the description of a law at finite resolution, then allow an average reconstruction error, and finally ask what information can be obtained from actual experiments.

If there were only finitely many candidate laws, naming one would require an index into that finite set. For a continuous class, the corresponding first step is to specify the law only to the desired physical resolution.

The two-law example in Figure~\ref{fig:orbit} is a useful first source. Suppose the inverse-square law $F_G$ and the linear restoring law $F_H$ are equally likely. Before observing an experiment, the learner lacks one binary choice. A circular trajectory does not settle that choice: both laws generate the same radial response. A release from rest can settle it because the responses differ. The uncertainty to be resolved is unchanged when the apparatus changes; what changes is the evidence available to resolve it. For a larger class, there are more choices, and sufficiently similar laws may be treated as equivalent at the requested accuracy.

Let $\F_0=\{F_w|_D:w\in\W\}\subseteq\F$ be the set of candidate physical fields on the evaluation domain. In distances and information quantities below, $F_w$ denotes this restriction. For $r>0$, let $N(\F_0,\loss,r)$ be the smallest number of closed radius-$r$ balls, centered in $\F_0$, that cover it. Its resolution-dependent metric entropy is
\begin{equation}\label{eq:entropy}
 \mathsf H_{\F_0}(r)=\log N(\F_0,\loss,r).
\end{equation}
All logarithms and information quantities use nats. Compactness makes this entropy finite at every positive resolution. A nearest-center description of the law has squared error at most $r^2$ and uses at most $\lceil\mathsf H_{\F_0}(r)/\log2\rceil$ bits. This is a description cost assuming the law is available to an encoder; it does not say those bits can be inferred from data. Neither network parameter count nor differential entropy of a continuous parameterization is substituted for~\eqref{eq:entropy}.

Write $\delta_{\rm law}=\loss(F_G,F_H)$ for the physical distance computed in~\eqref{eq:orbit-law-distance}. For this two-member class, $\mathsf H_{\F_0}(r)=\log2$ when $0<r<\delta_{\rm law}$ and equals zero when $r\ge\delta_{\rm law}$. The threshold is $\delta_{\rm law}$ because covering centers are required to be candidate laws. Below that radius, neither candidate can describe the other accurately enough. This simple example also shows why resolution must be declared: the same candidate class can require either one bit or no bits, depending on which discrepancies the task tolerates.

Let $P$ be a Borel prior on the joint world class, $W\sim P$, and $F=F_W$ the random target field. Allow an abstract reconstruction $\widetilde F\in\F$ through any measurable conditional distribution $Q(d\widetilde f\mid f)$. For random elements $X$ and $Y$ with joint law $P_{XY}$ and marginals $P_X,P_Y$, the mutual information is $I(X;Y)=\KL(P_{XY}\,\|\,P_X\otimes P_Y)\in[0,\infty]$. The law rate--distortion function is
\begin{equation}\label{eq:rd}
 \rd_P(\varepsilon)=
 \inf_{Q:\,\E_{P,Q}\loss(F,\widetilde F)^2\le\varepsilon}
 I_{P,Q}(F;\widetilde F).
\end{equation}
It is the least mutual information compatible with the requested average fidelity. It is nonincreasing in $\varepsilon$ and satisfies $\rd_P(r^2)\le\mathsf H_{\F_0}(r)$. The hypothetical output $\widetilde F$ in~\eqref{eq:rd} must be distinguished from the actual estimate $\lawhat=\decoder(T^\pi)$, whose information comes from an executable experiment.

Unlike the cover, the reconstruction in~\eqref{eq:rd} may lie outside the candidate class. Under the equal two-law prior, exact recovery requires $\rd_P(0)=\log2$ nats. However, the constant midpoint field $(F_G+F_H)/2$ has squared error $\delta_{\rm law}^2/4$ under either law and uses no observations. Thus $\rd_P(\varepsilon)=0$ for $\varepsilon\ge\delta_{\rm law}^2/4$. Guessing one of the two laws instead has selection error $1/2$ and average squared law error $\delta_{\rm law}^2/2$. These are different reconstruction rules and different losses. Rate--distortion measures the information needed for the declared field error, not necessarily the information needed to name a candidate correctly.

For the same prior, define the experimental information capacity
\begin{equation}\label{eq:capacity}
 \capinfo_{\A}(P,B)=\sup_{\pi\in\Pi_{\A}(B)}I_P(F;T^\pi)\in[0,\infty].
\end{equation}
This is a prior-dependent capacity at a finite acquisition budget, not a memoryless-channel capacity per use. It is nondecreasing in $B$. Although the transcript depends on both $F_W$ and $H_W$, the target of its information is $F$. Information about the entire world can also describe unknown sensing: $I_P(W;T^\pi)=I_P(F;T^\pi)+I_P(W;T^\pi\mid F)$. Counting that additional information as recovered physical law would change the question. For the equal prior on an equivalent pair with different laws, the capacity is zero at every budget. This is the information-theoretic form of the obstruction in Section~\ref{sec:distinguishability}.

\begin{samepage}
\begin{theorem}[Rate--distortion converse]\label{thm:rd}
Let Assumption~\ref{ass:standing} hold, let $B,\varepsilon\ge0$, and let $P$ be a Borel prior on $\W$. Suppose that a policy $\pi\in\Pi_{\A}(B)$ and a decoder $\decoder$ satisfy $\sup_{w\in\W}\E_w\loss(F_w,\decoder(T^\pi))^2\le\varepsilon$, and put $\lawhat=\decoder(T^\pi)$. Then
\begin{equation}\label{eq:rdchain}
 \boxed{\rd_P(\varepsilon)\le I_P(F;\lawhat)
       \le I_P(F;T^\pi)\le\capinfo_{\A}(P,B).}
\end{equation}
Consequently,
\begin{equation}\label{eq:rdbudget}
 B_{\A}^\star(\varepsilon)\ge
 \sup_P\inf\{B\ge0:\capinfo_{\A}(P,B)\ge\rd_P(\varepsilon)\},
\end{equation}
where the supremum is over Borel priors $P$ on $\W$ and $\inf\varnothing=+\infty$.
\end{theorem}
\end{samepage}
\begin{proof}
A uniform risk bound holds after averaging over $P$, so the conditional law of $\lawhat$ given $F$ is feasible in~\eqref{eq:rd}. The decoder only sees the transcript and independent randomness; hence $F\to T^\pi\to\lawhat$ is a Markov chain. Data processing gives the middle inequality, and~\eqref{eq:capacity} gives the last. Every feasible budget satisfies the chain for every prior, yielding~\eqref{eq:rdbudget}.
\end{proof}

For the circular-orbit apparatus, the chain has a transparent interpretation. The transcript carries zero information about the binary law, so every estimate computed from it also carries zero information. The best constant reconstruction attains the midpoint floor, but an error below that floor requires a distinction absent from the data. If release experiments are admitted, the observations carry information about that same binary choice; independent repetitions can make the choice increasingly reliable. The theorem applies equally to a symbolic estimator and a learned world model because neither can increase the information about the target after acquisition. It does not prescribe the internal representation used to retain that information.

The theorem is Shannon's distortion-constrained information converse applied to physical experiments. Metric entropy makes its implication concrete. For $r>0$, let $M(\F_0,\loss,2r)$ be the largest cardinality of a subset of $\F_0$ whose elements have pairwise distances strictly greater than $2r$; compactness of $\F_0$ makes it finite. Nearest-law decoding and Fano's inequality give the following bound.
\begin{corollary}[Packing lower bound]\label{cor:packing}
Let Assumption~\ref{ass:standing} hold, let $r>0$, and suppose $M=M(\F_0,\loss,2r)\ge2$. Let $f_1,\ldots,f_M\in\F_0$ have pairwise distances strictly greater than $2r$, choose $w_j\in\W$ with $F_{w_j}=f_j$, and let $P_r$ be the uniform prior on $\{w_1,\ldots,w_M\}$. Then, for every $\varepsilon\ge0$,
\begin{equation}\label{eq:packing}
 \rd_{P_r}(\varepsilon)
 \ge(1-\varepsilon/r^2)\log M-\log2.
\end{equation}
If, in addition, some $\pi\in\Pi_{\A}(B)$ and decoder $\decoder$ have uniform squared law risk at most $\varepsilon$, then $\capinfo_{\A}(P_r,B)\ge(1-\varepsilon/r^2)\log M-\log2$. In particular, if $\varepsilon>0$ and $M(\F_0,\loss,4\sqrt\varepsilon)\ge2$, such a budget satisfies $\capinfo_{\A}(P_{2\sqrt\varepsilon},B)\ge\tfrac34\log M(\F_0,\loss,4\sqrt\varepsilon)-\log2$.
\end{corollary}
The bound is informative only when $\varepsilon<r^2$. For example, set $r=2\sqrt\varepsilon>0$ and suppose the resulting packing contains at least two laws. The required information is then at least $\tfrac34\log M(\F_0,\loss,4\sqrt\varepsilon)-\log2$ nats. Entropy and information thus connect a physical accuracy requirement to the distinctions acquisition must resolve, as in information-theoretic minimax estimation \citep{yang1999}. Appendix~\ref{app:information} proves Corollary~\ref{cor:packing} and supplies Gaussian information-per-cost bounds.

This inequality is a necessary condition. Exceeding the scalar information requirement does not ensure recoverability. A sensor may reveal several bits about a small component of the field while hiding a single bit that controls most of its physical error; Appendix~\ref{app:information} gives a finite example with $\capinfo_{\A}(P,B)>\rd_P(\varepsilon)$ but unattainable distortion $\varepsilon$. Unlike a transmitter that knows and encodes a source, the experimenter only probes an unknown law through a constrained apparatus. A recovery theorem must therefore exhibit experiments that resolve the relevant physical distinctions. This is the role of the next section.

\section{From Distinguishable Responses to Recovered Laws}\label{sec:recovery}
The converse asks how much information recovery requires. We now construct an experimental program and a decoder that suffice. The construction treats a world as a joint hypothesis about dynamics and sensing, because the same physical field can produce different responses under different sensors. Its output is only the physical field. The idea is to list finitely many representative responses, repeat a program that separates physically different laws, and return the field whose representative response best matches the observed average.

For Figure~\ref{fig:orbit}, this is simply a comparison between two predicted release radii. Prepare the object at $\boldsymbol q=(2,0)$ with $\boldsymbol v=0$, observe its radius at $t_*=0.4$, and compare the average reading with the two candidate predictions. Let $r_G(t)$ and $r_H(t)$ denote the release radii predicted by $F_G$ and $F_H$. The relevant experimental separation is the response gap $d_{\rm rel}=|r_G(t_*)-r_H(t_*)|$, whereas the relevant reconstruction loss is the field distance $\delta_{\rm law}$ on the evaluation domain. They answer different questions: the first determines how difficult the measurement is; the second determines how costly an incorrect conclusion would be. The construction below keeps these two distances separate for an arbitrary compact world class.

To compare observations with different noise covariances, transform each response to $Z_a=\Sigma_a^{-1/2}Y_a$ and write its mean as $\mu_a(w)=\Sigma_a^{-1/2}m_a(w)$. This known invertible transformation, called whitening, gives identity noise covariance and preserves all information. For a finite nonempty program $A\subset\A$, define its stacked mean $\mu_A(w)=(\mu_a(w))_{a\in A}$, output dimension $Q_A=\sum_{a\in A}k_a$, and one-round cost $\cost(A)=\sum_{a\in A}c_a$. After $n\ge1$ independent rounds, with $n$ an integer, the average response is
\begin{equation}\label{eq:average}
 \overline Z_A=\mu_A(w)+G/\sqrt n,\qquad G\sim N(0,I_{Q_A}),
\end{equation}
at cost $n\cost(A)$. At a physical resolution $0<r\le D_F$, let
\begin{equation}\label{eq:sep}
 \Delta_A(r)=\min_{\loss(F_w,F_v)\ge r}
        \|\mu_A(w)-\mu_A(v)\|_2.
\end{equation}
The minimum is over joint worlds $w,v\in\W$; under Assumption~\ref{ass:standing} the constraint set is nonempty and compact and $\mu_A$ is continuous, so the minimum is attained. A positive gap means that every pair of laws whose confusion matters at resolution $r$ is separated in response space, even after allowing their sensors to differ.

In the two-law release example, $A$ contains one scalar readout with noise standard deviation $\sigma$, so $\Delta_A(r)=d_{\rm rel}/\sigma$ for $0<r\le\delta_{\rm law}$. Whitening therefore expresses the radius gap in units of observational noise. The circular-orbit program instead has $\Delta_A(r)=0$ throughout this range. Repeating a separating program reduces the noise of its mean by $\sqrt n$; repeating a zero-gap program never makes its two means different. This is the operational distinction between acquiring more data and acquiring a new kind of evidence.

For $h>0$, a codebook of radius $h$ is a finite family of representative worlds $v_1,\ldots,v_{M}\in\W$ whose responses form an $h$-cover of the compact set $\mu_A(\W)$:
\begin{equation}\label{eq:codebook}
 \sup_{w\in\W}\min_{1\le j\le M}
 \|\mu_A(w)-\mu_A(v_j)\|_2\le h.
\end{equation}
Compactness of $\mu_A(\W)$ guarantees that such codebooks exist. Let $M_A(h)\ge1$ denote the smallest size of a codebook of radius $h$. It includes response variability due to sensing, so its logarithm is not the law entropy of~\eqref{eq:entropy}. The response-matching decoder is
\begin{equation}\label{eq:decoder}
 \widehat j=\min\argmin_{1\le j\le M}
 \|\overline Z_A-\mu_A(v_j)\|_2,
 \qquad \lawhat=F_{v_{\widehat j}}.
\end{equation}
The smallest-index rule resolves ties; the true world need not occur in the codebook.

For the running example, take the two candidate worlds themselves as codewords. Their means are represented exactly, and nearest-response matching places the decision threshold halfway between the predicted radii. The probability of selecting the wrong law under either candidate is $p_e(n)=\Phi(-\sqrt n\,d_{\rm rel}/(2\sigma))$, where $\Phi$ is the standard normal distribution function. The selected field then has expected squared law error $\delta_{\rm law}^2p_e(n)$. This is an exact risk for this particular two-candidate decoder, not a claim that candidate selection minimizes squared error over all field-valued outputs. A posterior mean, for example, may interpolate between the two fields. The theorem gives a uniform guarantee without requiring such a closed-form binary calculation.

\begin{theorem}[Constructive law recovery]\label{thm:recovery}
Let Assumption~\ref{ass:standing} hold, let $A\subset\A$ be a finite nonempty program, and let $n\ge1$ be an integer. Let $0<r\le D_F$ satisfy $\Delta_A(r)>0$, let $0<h<\Delta_A(r)$, and let $v_1,\ldots,v_M$ be a codebook of radius $h$ with $M=M_A(h)$. Applied to the averaged response~\eqref{eq:average} of $n$ independent rounds of $A$, the decoder~\eqref{eq:decoder} satisfies
\begin{align}
 \sup_{w\in\W}\Pr_w\{\loss(F_w,\lawhat)\ge r\}
 &\le\min\!\left\{1,M_A(h)e^{-n(\Delta_A(r)-h)^2/8}\right\},\label{eq:decode_error}\\
 \sup_{w\in\W}\E_w\loss(F_w,\lawhat)^2
 &\le r^2+D_F^2M_A(h)e^{-n(\Delta_A(r)-h)^2/8}.
 \label{eq:decode_risk}
\end{align}
Both bounds remain valid when $M_A(h)$ is replaced by the size of any codebook of radius $h$. For $0<\varepsilon<D_F^2$, if $\Delta_A(\sqrt{\varepsilon/2})>0$, setting $r=\sqrt{\varepsilon/2}$, $h=\Delta_A(r)/2$, and $n=\lceil 32\Delta_A(r)^{-2}\log(2M_A(h)D_F^2/\varepsilon)\rceil$ yields uniform risk at most $\varepsilon$ at deterministic cost $n\cost(A)$; hence
\begin{equation}\label{eq:constructive_budget}
 B_{\A}^\star(\varepsilon)\le
 \cost(A)\left\lceil\frac{32}{\Delta_A(r)^2}
    \log\frac{2M_A(\Delta_A(r)/2)D_F^2}{\varepsilon}\right\rceil.
\end{equation}
\end{theorem}

To see the mechanism, compare a physically wrong candidate with a codeword whose mean is within $h$ of the truth. The wrong candidate is at least $\Delta_A(r)$ away in mean. It can have the smaller residual only through a one-dimensional Gaussian fluctuation with probability at most $e^{-n(\Delta_A(r)-h)^2/8}$. A union bound over the codebook proves~\eqref{eq:decode_error}; bounded target diameter gives~\eqref{eq:decode_risk}. Appendix~\ref{app:recovery} proves the result, including off-codebook truth. The sufficient budget separates three costs: the number of candidate responses, the precision needed to distinguish them, and the price of executing the program.

A small response codebook does not automatically mean a simple physical-law class, and a large one can reflect nuisance variation in the sensor. Its purpose is to approximate what the apparatus may observe. Positive separation ensures that matching those observations is enough to control the physical error. Constructing this codebook may nevertheless be computationally difficult: the result establishes an acquisition guarantee, and the budget does not charge for enumerating worlds, solving their forward dynamics, or searching representatives. It is therefore an existence construction against which a practical learning procedure can be assessed, rather than a prescribed scalable training algorithm.

The theorem also yields an exact answer to when recovery is possible. For a finitely supported prior $P$, let $\mathsf H_P(F)$ denote the ordinary Shannon entropy of its finite random target law, counting identical fields only once: if $f_1,\ldots,f_m$ are the distinct values of $F=F_W$ under $P$, then $\mathsf H_P(F)=-\sum_{i=1}^m P(F=f_i)\log P(F=f_i)$.
\begin{corollary}[Experimental completeness]\label{cor:characterization}
Under Assumption~\ref{ass:standing}, the following are equivalent: (i) $\kappa_{\A}=0$; (ii) for every finitely supported prior $P$, $\lim_{B\to\infty}\capinfo_{\A}(P,B)=\mathsf H_P(F)$; (iii) $\lim_{B\to\infty}\mathfrak R_{\A}^\star(B)=0$.
\end{corollary}
Indeed, if every distinct-law pair is separated, compactness supplies a finite program $A_r$ with $\Delta_{A_r}(r)>0$ at each $0<r\le D_F$. Theorem~\ref{thm:recovery} then gives a finite budget at every tolerance. Repetitions also identify the law of any finite source with vanishing error, recovering all its entropy. An equivalent distinct-law pair violates both claims. The quantifiers matter: a program may change as the requested resolution becomes finer, and completeness means resolving every finite law source, not merely exceeding one scalar information threshold.

The remaining issue is stability. For a fixed program, define the inverse response modulus
\begin{equation}\label{eq:modulus}
 \omega_A(s)=\max_{\|\mu_A(w)-\mu_A(v)\|_2\le s}\loss(F_w,F_v),\qquad s\ge0.
\end{equation}
It is the largest law discrepancy compatible with response uncertainty $s$. For an integer $n\ge1$, let $Z^{(1)},\ldots,Z^{(n)}$ be independent whitened responses of $n$ complete rounds of $A$, so that $Z^{(i)}\sim N(\mu_A(w),I_{Q_A})$, and define the minimax squared law risk
\[
 \mathfrak R_A(n)=\inf_{\widehat F}\sup_{w\in\W}\E_w\loss\bigl(F_w,\widehat F(Z^{(1)},\ldots,Z^{(n)})\bigr)^2,
\]
where the infimum is over all measurable, possibly randomized, estimators with values in $\F$.
\begin{proposition}[Inverse response modulus]\label{prop:modulus}
Let Assumption~\ref{ass:standing} hold and let $A\subset\A$ be a finite nonempty program. The maximum in~\eqref{eq:modulus} is attained, and $\omega_A$ is nondecreasing and right-continuous with values in $[0,D_F]$. For every integer $n\ge1$,
\begin{equation}\label{eq:fixedmodulus}
 \tfrac12\Phi(-\tfrac12)\omega_A(n^{-1/2})^2
 \le\mathfrak R_A(n)
 \le\E\,\omega_A(2\|G\|_2/\sqrt n)^2,
\end{equation}
where $G\sim N(0,I_{Q_A})$ and $\Phi$ is the standard normal distribution function. If, moreover, there are constants $0<c\le C$, $s_0>0$, and $0<\beta\le1$ such that $cs^\beta\le\omega_A(s)\le Cs^\beta$ for $0\le s\le s_0$, then there are constants $0<c'\le C'$ and an integer $n_0$ such that $c'n^{-\beta}\le\mathfrak R_A(n)\le C'n^{-\beta}$ for all $n\ge n_0$.
\end{proposition}
This is the inverse-problem modulus principle \citep{donoho1994,donoho1991}; Appendix~\ref{app:recovery} proves Proposition~\ref{prop:modulus}. Identifiability alone gives no such rate. This modulus concerns one fixed experimental program: it describes how noise in that program's responses becomes uncertainty about the field. Optimizing across programs is a further problem, since a better-separated program may also cost more, and finer resolution may require additional preparations or sensors. The next section computes the optimal cost for a concrete infinite-dimensional class, beyond the general codebook construction.

\section{The Experimental Cost of a Physical Law}\label{sec:budget}
The previous sections separate information requirements from experimental recoverability. The release example has only two possible laws. Once their responses differ, repeated Gaussian measurements make candidate selection increasingly reliable with a Gaussian-tail error. Recovering an arbitrary field is a different task: improving the requested accuracy exposes new spatial distinctions, so both the number and the precision of experiments must increase. We now compute the minimax acquisition cost for an infinite-dimensional law class. Let $p\ge1$, $q\ge0$, and $d=p+q$. For fixed $L,M,\sigma>0$, let
\begin{equation}\label{eq:lipschitzclass}
 \F_{L,M}=\{F:\R^p\times[0,1]^q\to\R^p:
 \sup_z\|F(z)\|_2\le M,\ \|F(z)-F(z')\|_2\le L\|z-z'\|_2\}.
\end{equation}
Here $z=(x,u)$ denotes a state--action pair, and $D=[0,1]^d$ carries the uniform measure in~\eqref{eq:metric}. Fields are defined beyond $D$ because trajectories may leave the evaluation domain. This class is compact under uniform convergence on compact subsets of its domain.

The catalog $\A_{\rm state}$ resets exactly to any known $x_0\in[0,1]^p$, holds any constant $u\in[0,1]^q$, and records the full physical state at $k\ge1$ times $0<\tau_1<\cdots<\tau_k$. Each reading is $x(\tau_i)+\xi_i$ with independent $\xi_i\sim N(0,\sigma^2I_p)$. The sensor is known. An episode costs $c_{\rm r}+k+c_{\rm t}\tau_k$, with fixed reset and time costs $c_{\rm r},c_{\rm t}\ge0$; every readout costs one unit. There are no observations between paid readouts. Write $B^\star_{\rm state}(\varepsilon):=B^\star_{\A_{\rm state}}(\varepsilon)$ for the least feasible budget defined in~\eqref{eq:budgetfrontier}.

To visualize this apparatus, take a scalar system with no control, so $p=d=1$. The unknown curve $F(x)$ specifies how fast the state moves at every $x\in[0,1]$. A reset chooses a starting point on this interval; a later noisy state reading estimates the short displacement from that point. Sampling only one region leaves the rest of the curve insufficiently constrained under the Lipschitz assumption. A successful design must therefore combine coverage of the interval with accurate local measurements. In higher dimensions the same requirement applies across both state and action coordinates.

The exponent below agrees with the nonparametric ODE rates of \citet{schoetz2026,schoetz_siebel2026}. We establish a uniform bound over sequential acquisition policies for a Lipschitz class, without a sampling-density or trajectory-coverage condition on the selected experiments. In particular, the policy may concentrate its measurements based on earlier observations.

\begin{theorem}[Minimax acquisition cost]\label{thm:lipschitz-budget}
Let $p\ge1$, $q\ge0$, $d=p+q$, $L,M,\sigma>0$, and $c_{\rm r},c_{\rm t}\ge0$. Let the world class be $\F_{L,M}$ in~\eqref{eq:lipschitzclass}, each world observed through the known identity sensor, with evaluation domain $D=[0,1]^d$ under the uniform measure and catalog $\A_{\rm state}$. There are $c_1,c_2>0$ depending only on $d,p$, and $\varepsilon_0>0$ depending only on $d,p,L,M,\sigma$, such that for $0<\varepsilon\le\varepsilon_0$,
\begin{equation}\label{eq:statebudget}
 c_1\sigma^2L^{d+2}M^2\varepsilon^{-(d+4)/2}
 \le B^\star_{\rm state}(\varepsilon)
 \le c_2(1+c_{\rm r}+c_{\rm t})\sigma^2L^{d+2}M^2\varepsilon^{-(d+4)/2}.
\end{equation}
A nonadaptive finite-difference design attains the upper bound. The lower bound applies to every adaptive policy with pathwise cost at most $B$ and every measurable, possibly randomized, estimator, and remains valid if readout times and stopping within an episode adapt to earlier paid readouts, while its action remains constant.
\end{theorem}

The upper bound balances three errors: approximating the field across the domain, approximating a derivative by a short trajectory, and averaging measurement noise. For the upper bound, divide $D$ into cubes of side $1/m$ with centers $z_i=(x_i,u_i)$. At every center perform $n$ independent resets and take one readout after time $t$. Let $\overline Y_i$ denote their mean. With $\operatorname{proj}_M$ the Euclidean projection onto the radius-$M$ ball, return
\begin{equation}\label{eq:fd}
 \lawhat(z)=\operatorname{proj}_M\!\left(\frac{\overline Y_i-x_i}{t}\right)
 \quad\text{for }z\text{ in cube }i.
\end{equation}
The finite difference has deterministic error at most $LMt/2$, since the field is $L$-Lipschitz and the state moves at speed at most $M$. Its variance is $p\sigma^2/(nt^2)$. Spatial interpolation adds at most $L\sqrt d/(2m)$ in norm. Thus $m\asymp L/\sqrt\varepsilon$, $t\asymp\sqrt\varepsilon/(LM)$, and $n\asymp\sigma^2L^2M^2/\varepsilon^2$ suffice. Multiplying the number of centers by the repetitions and the per-episode cost gives the upper bound.

The short-time choice is essential. Waiting longer increases the displacement relative to the readout noise, but also makes that displacement depend on the field along a longer trajectory rather than only at the reset state. Waiting less reduces this differentiation bias while magnifying noise after division by $t$. Repetition compensates for that amplified noise. Finer spatial cells alone do not remove temporal bias, and longer trajectories alone do not supply uniformly accurate local derivatives. The lower bound explains why a different adaptive choice of the admitted experiments cannot improve the worst-case exponent.

The difficulty in the lower bound is that a long trajectory could cross many local perturbations before it is read. We place sign-changing perturbations on a constant positive drift. Reversing a perturbation preserves its exact full-crossing time, by reflection symmetry, so the trajectories rejoin at its exit. A readout is affected by at most two sign choices: the cell containing its reset and the cell it currently occupies. Summed Gaussian divergences therefore remain small regardless of the duration. An adaptive change-of-measure argument and a hypercube testing reduction give the lower bound. Appendix~\ref{app:lipschitz-budget} provides the full construction and constants.

For comparison, let $\A_{\rm vel}$ return $F(x_0,u)+N(0,\sigma^2I_p)$ directly at cost $1+c_{\rm r}$, for any $x_0\in[0,1]^p$ and $u\in[0,1]^q$, with independent noise across readings, and write $B^\star_{\rm vel}=B^\star_{\A_{\rm vel}}$. The same spatial approximation needs only $n\asymp\sigma^2/\varepsilon$ readings per center.
\begin{proposition}[Direct velocity readouts]\label{prop:velocity}
Under the conditions of Theorem~\ref{thm:lipschitz-budget}, with $\A_{\rm vel}$ in place of $\A_{\rm state}$, there are $c_1',c_2'>0$ depending only on $d,p$, and $\varepsilon_0'>0$ depending only on $d,p,L,M,\sigma$, such that for $0<\varepsilon\le\varepsilon_0'$,
\begin{equation}\label{eq:velocitybudget}
 c_1'\sigma^2L^d\varepsilon^{-(d+2)/2}
 \le B^\star_{\rm vel}(\varepsilon)
 \le c_2'(1+c_{\rm r})\sigma^2L^d\varepsilon^{-(d+2)/2}.
\end{equation}
The lower bound holds for every adaptive policy and estimator.
\end{proposition}
For $d=1$, these bounds correspond to squared errors of order $B^{-2/5}$ for state readouts and $B^{-2/3}$ for direct field readouts, with problem constants fixed. Halving the root-mean-square error means quartering $\varepsilon$: the corresponding power-law budget factors are $2^{d+4}$ and $2^{d+2}$, or $32$ and $8$ in this scalar example, with multiplicative problem constants suppressed. These are worst-case accuracy costs for the whole Lipschitz class, not the binary error curve for the two known orbit laws.

This is the Lipschitz regression rate \citep{stone1982}; Appendix~\ref{app:lipschitz-budget} proves Proposition~\ref{prop:velocity} with explicit constants. The extra factor $\varepsilon^{-1}$ for state readouts is a worst-case acquisition penalty in this model. Parametric restrictions, greater smoothness, direct velocity sensing, or additional control operations can change it. In particular, a second-order mechanical system observed only through positions is a different observation model; here the entire first-order state is measured. For the orbit system this would require observing both position and velocity as components of the first-order state; its scalar radial sensor does not meet this assumption. The orbit example establishes the value of an informative intervention, while the present theorem isolates a separate cost of recovering a general field from full-state measurements.

\section{Intervention, Calibration, and Measurement}
\label{sec:cases}

We use five studies to make the theory operational. The first two separate informative preparation from physical calibration. Two further orbital studies examine the choice of release and readout, then recover a continuous force parameter using a finite response codebook. The final study examines the cost of local field reconstruction from noisy observations. All reported numerical results were generated from the accompanying source. These experiments illustrate the mathematical mechanisms; they do not compare neural architectures or empirically establish minimax optimality.

We first return to the two laws in Figure~\ref{fig:orbit}. Releasing the object from rest at radius two produces different position responses. At time $t=0.4$ the radial responses differ by $\Delta=0.13781$. With independent Gaussian radial readout noise of standard deviation $\sigma=0.1$, the optimal two-candidate test under equal priors after $n$ releases has error $\Phi(-\sqrt n\Delta/(2\sigma))$, where $\Phi$ is the standard normal distribution function. It decreases from $0.2454$ at $n=1$ to $0.002924$ at $n=16$, whereas repetition of the circular orbit leaves the testing error at $1/2$ for every budget. Figure~\ref{fig:cases}(a) compares the exact curve with 200,000 Monte Carlo trials per budget. Repetition suppresses noise only after an experiment has exposed a distinction.

Unknown sensors create a different obstruction. Consider the nonlinear family
\begin{equation}
\dot x=-ax-bx^3+g(u+u^3),\qquad y=cx+b_0+\xi.
\label{eq:nonlinear-sensor-case}
\end{equation}
Here $a,b,g$ are physical coefficients, $c>0$ is sensor gain, $b_0$ is sensor offset, and $\xi$ is readout noise. If episodes start at $x=0$, the transformation
\[
(a,b,g,c,b_0)\longmapsto(a,b/\lambda^2,\lambda g,c/\lambda,b_0),\qquad \lambda>0,
\]
maps the physical trajectory to $\lambda x(t)$ while preserving every sensor response. This holds for any admissible input waveform, and therefore also for input choices based on previous observations, whenever both worlds belong to the declared class. Excitation alone cannot determine the physical scale. A known-state reference at $x=1$ supplies a response $c+b_0$ that changes under this transformation; a reference at $x=0$ separately measures $b_0$.

To isolate this obstruction numerically, we use a separate linear gain family $\dot x=-x+gu$, $y=cx+\xi$, with $cg=1$ and $c\sim\mathrm{Unif}[0.75,1.25]$. For the noiseless output coordinate $z=cx$, the observed dynamics are $\dot z=-z+u$ for every world, so even exact dynamical responses contain no information about the remaining gain. Under uniform evaluation on $(x,u)\in[-1,1]^2$, squared law error is $(\hat g-g)^2/3$. The optimal uncalibrated Bayes decoder has error $\mathrm{Var}(1/c)/3=0.0076318$ at every budget. After $n$ independent known-state references with mean reading $\bar y$, the clipped decoder $\hat c=\operatorname{clip}_{[0.75,1.25]}(\bar y)$, $\hat g=1/\hat c$ has vanishing error. Each reference costs one clamp and one readout, so $B=2n$. With $\sigma=0.08$, 1,024 worlds and 128 noise replicates per world, mean law error falls from $2.04\times10^{-3}$ at $B=2$ to $2.52\times10^{-6}$ at $B=2048$; the last three budgets give a slope of $-0.994$. Noiseless dynamic-output prediction remains unchanged because $\hat c\hat g=1$. Figure~\ref{fig:cases}(b) shows the distinction between accurate observable predictions and calibrated physical-law recovery.

\begin{figure}[!htbp]
\centering
\includegraphics[width=\linewidth]{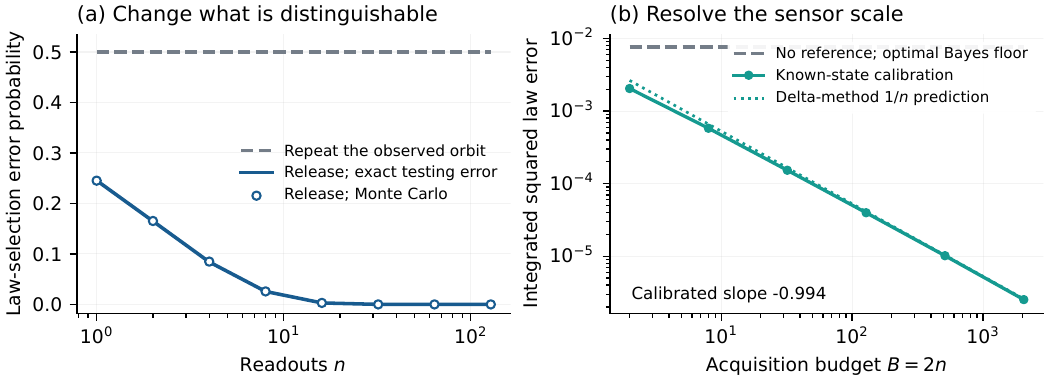}
\caption{Evidence that changes recoverability. Left: exact and simulated two-law selection error after release; repeating the circular orbit remains uninformative. Right: the exact uncalibrated Bayes floor and calibrated squared law error in the sensor-gain family. Shading is $\pm1.96$ standard errors over world averages; the dotted curve is the small-noise approximation. The panels report different risks, as labeled.}
\label{fig:cases}
\end{figure}

We now return to the calibrated radial sensor in the orbital example and ask a more quantitative question: which release experiment makes the two laws easiest to distinguish at a given cost? Let $R\in[1,2]$ be the release radius, with zero initial velocity, and let $t\in[0.05,0.6]$ be the time of a single radial reading. Denote the predicted means by $m_G(R,t)$ and $m_H(R,t)$. Each independent episode costs one reset, one readout, and its elapsed time, so $c(R,t)=2+t$. With noise standard deviation $\sigma=0.1$, a useful pairwise design score is
\begin{equation}\label{eq:release-design-score}
 S(R,t)=\frac{|m_G(R,t)-m_H(R,t)|^2}{\sigma^2(2+t)}.
\end{equation}
This is response separation per unit cost, rather than the mutual-information capacity of~\eqref{eq:capacity}. The exact equal-prior testing error at budget $B$ is obtained from the binary formula using $n=\lfloor B/(2+t)\rfloor$ independent releases. When the mean gap is positive, the least number of repetitions of this fixed episode giving error at most $0.01$ is
\begin{equation}\label{eq:release-onepercent}
 n_{0.01}(R,t)=\left\lceil
 \left(\frac{2\sigma\Phi^{-1}(0.99)}{|m_G(R,t)-m_H(R,t)|}\right)^2
 \right\rceil.
\end{equation}

Figure~\ref{fig:release-design} compares a grid of 101 release radii and 111 readout times. Table~\ref{tab:release-designs} gives several individual designs. At $R=2$, delaying the readout from $0.2$ to $0.6$ increases the cost of one episode from $2.2$ to $2.6$, but reduces the number of required releases from 179 to 3. The total cost falls from $393.8$ to $7.8$. Releasing from radius one at rest also eventually separates the laws, even though their initial accelerations agree; this preparation differs from the tangential velocity that sustains the original circular orbit. Its much smaller response gap makes it expensive at the displayed time. These comparisons concern specified two-law tests. The highest score on the displayed grid occurs at $(R,t)=(2,0.6)$; it is not a claim of globally optimal experimental design. Integer rounding in~\eqref{eq:release-onepercent} also means that maximizing the score need not minimize the exact finite-budget threshold.

\begin{table}[!htbp]
\centering\small
\caption{Cost of distinguishing the two orbit laws with error at most $1\%$. Each row repeats the same reset and one-readout episode, with $\sigma=0.1$. The cost is $B_{0.01}=(2+t)n_{0.01}$. These are analytically evaluated Gaussian tests, not fitted learning curves.}
\label{tab:release-designs}
\begin{tabular}{@{}rrrrr@{}}
\toprule
Release $R$ & Readout $t$ & Mean gap & Releases $n_{0.01}$ & Cost $B_{0.01}$\\
\midrule
1.0 & 0.6 & 0.017865 & 679 & 1765.4\\
1.5 & 0.4 & 0.082567 & 32 & 76.8\\
2.0 & 0.2 & 0.034863 & 179 & 393.8\\
2.0 & 0.4 & 0.137811 & 12 & 28.8\\
2.0 & 0.6 & 0.303986 & 3 & 7.8\\
\bottomrule
\end{tabular}
\end{table}

\begin{figure}[!htbp]
\centering
\includegraphics[width=\linewidth]{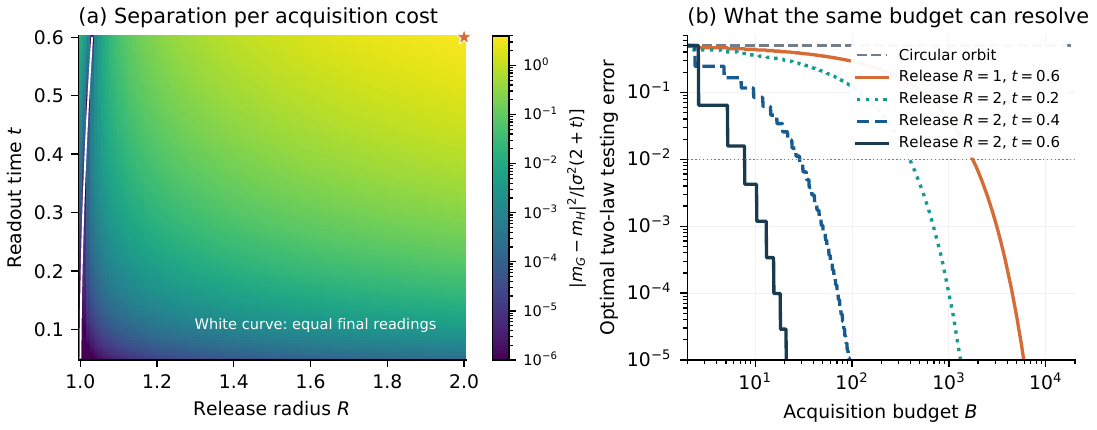}
\caption{The same two physical laws can be easy, difficult, or impossible to distinguish under different readout designs. (a) The horizontal axis changes the release radius and the vertical axis changes the readout time. Color shows the separation-per-cost score~\eqref{eq:release-design-score} on a logarithmic scale; brighter regions give more separation for the charged cost. The star marks the largest score on the tested grid. The white curve is a set of equal terminal readings, although the intervening trajectories differ. (b) Exact Gaussian testing errors at the same total budget, including the integer number of affordable episodes. The gray dashed line is the uninformative circular preparation, and the horizontal dotted line marks $1\%$ error. Every curve uses the same known sensor noise $\sigma=0.1$.}
\label{fig:release-design}
\end{figure}

The white curve reveals a second kind of ambiguity within this example. At $t=0.6$, the signed mean difference is negative at $R=1$ and positive at $R=1.1$. Continuity therefore supplies a radius $R_*$ with equal terminal means; numerical solution gives $R_*\simeq1.030748$. Repeating only that final reading cannot distinguish the laws, although the full trajectories are different. An additional reading at $t=0.3$ has mean gap approximately $0.00308325$ and restores distinguishability. It remains a weak design: with independent noise of standard deviation $0.1$ on both readings, reaching $1\%$ error requires 22,772 releases and cost $81,979.2$. This illustrates the difference between removing an exact ambiguity and obtaining a stable, economical experiment. Appendix~\ref{app:experiments} states the root-existence argument and the two-readout calculation.

The binary examples assume that the true law is one of two stored hypotheses. To examine the response-codebook construction when the truth is not itself stored, consider the continuous family
\begin{equation}\label{eq:continuous-force-family}
 a_\theta(\boldsymbol q)=-(1-\theta)\boldsymbol q
 -\theta\frac{\boldsymbol q}{\|\boldsymbol q\|^3},
 \qquad F_\theta(\boldsymbol q,\boldsymbol v)
 =(\boldsymbol v,a_\theta(\boldsymbol q)),\qquad 0\le\theta\le1.
\end{equation}
Every member generates the same unit circular orbit with unit tangential speed. Thus the original observational ambiguity is now a continuum, not merely a pair. We release from $R=2$ at rest and read the known radial sensor at $t=0.6$. Let $m(\theta)$ denote this noiseless response. A codebook stores the responses at $K_c\in\{17,65,257\}$ equally spaced parameter values, and the decoder selects the stored response nearest to the observed sample mean. The true parameters are 32 fixed values, all absent from every codebook; they include $0.37$ and $0.63$, followed by 30 uniform draws retained across all comparisons. This is a fixed-truth comparison, not a numerical uniform-risk guarantee over $[0,1]$.

The loss remains the physical field error of Section~\ref{sec:problem}. Linearity in the force parameter gives the exact identity
\begin{equation}\label{eq:continuous-force-loss}
 \loss(F_\theta,F_{\widehat\theta})^2
 =\delta_{\rm law}^2(\theta-\widehat\theta)^2,
 \qquad \delta_{\rm law}^2=21/8-2\log2.
\end{equation}
Figure~\ref{fig:codebook-study} separates three operations: constructing the response codebook, refining that codebook, and improving the measurement precision. At each configuration we simulate 50,000 Gaussian sample means per true parameter and compute the decoder's noise-averaged error independently from Gaussian probabilities at its response thresholds. With $\sigma=0.1$ and $n=16,384$ releases, the mean squared law errors for 17, 65, and 257 entries are approximately $3.909\times10^{-4}$, $3.452\times10^{-5}$, and $9.867\times10^{-6}$, respectively. The cost is $B=2.6n=42,598.4$. Over the 40 compared designs, aggregate Monte Carlo estimates differ from the probability calculations by at most 3.48 estimated standard errors.

The horizontal dotted levels are the errors of the respective decoders given a noiseless response, averaged over the same 32 true parameters. They are approximately $3.834\times10^{-4}$, $2.491\times10^{-5}$, and $1.308\times10^{-6}$. These levels arise from the fixed finite codebooks. They are not the apparatus ambiguity floor $\kappa_{\A}^2/4$ and do not bound an estimator allowed to return intermediate parameter values. More repetitions reduce observational noise; a finer codebook reduces the remaining approximation error. Holding the codebook at 65 entries and lowering $\sigma$ from $0.3$ to $0.03$ reduces the error at the largest budget from approximately $1.000\times10^{-4}$ to $2.540\times10^{-5}$, close to that codebook's noiseless level. The example makes the separate roles of response separation, finite approximation, and measurement noise visible within the same physical family.

\begin{figure}[!htbp]
\centering
\includegraphics[width=\linewidth]{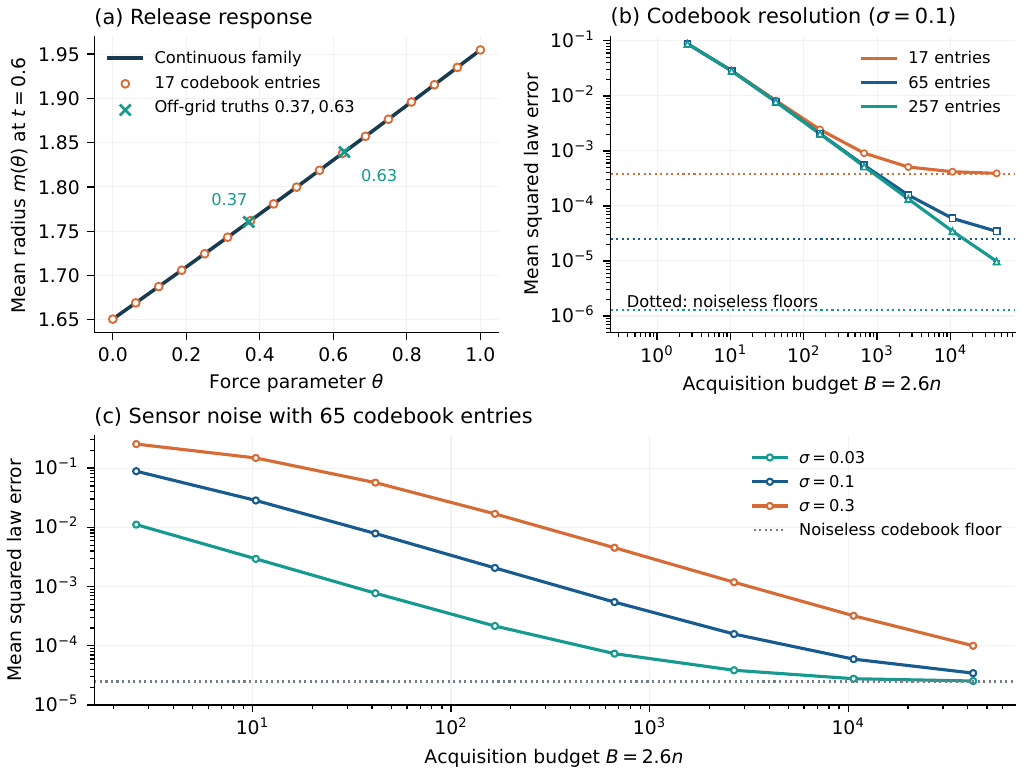}
\caption{Recovering a continuous force parameter when the true world is not in the response codebook. (a) The release response $m(\theta)$ varies with the force parameter. Orange circles are the 17 coarse codewords; teal crosses illustrate two off-codebook truths. All these laws agree on the original circular orbit. (b) Noise-averaged squared field error versus acquisition budget, for three codebook sizes at $\sigma=0.1$. Lines use Gaussian decision probabilities; markers and bars show Monte Carlo means and $\pm1.96$ standard errors, conditional on the same 32 true parameters. The bars may be smaller than the markers. Dotted lines show each decoder's noiseless approximation error. (c) The same error for different sensor noise levels with 65 entries. Improving precision approaches the fixed codebook's approximation level. These are finite-decoder studies in a known one-parameter family, rather than an empirical minimax result or a comparison of neural architectures.}
\label{fig:codebook-study}
\end{figure}

The continuous force family still has only one unknown parameter. The final study returns to the different problem in Section~\ref{sec:budget}: estimating a spatially varying field by combining coverage of the evaluation region with repeated short-time measurements.

Figure~\ref{fig:scaling} evaluates the unprojected finite-difference construction of~\eqref{eq:fd} on two bounded Lipschitz fields with state--action dimensions one and two. With readout noise $\sigma=0.05$, a grid with $m$ cells per coordinate uses duration $t=1/m$ and $n=\lceil\sigma^2m^4\rceil$ independent state readouts per cell. The velocity comparator uses $n=\lceil\sigma^2m^2\rceil$; the fixed-duration comparator keeps $t=0.25$. Reset, readout, and elapsed time are all charged. The Gaussian sample mean is sampled directly at its exact variance. Exact piecewise-affine flows and analytic cell integrals remove integration and evaluation-grid errors from the comparison. Across 64 independent noise replicates, the fitted state-readout slopes are $-0.400$ and $-0.332$, and the velocity-readout slopes are $-0.667$ and $-0.488$, close to the design exponents $-2/(d+4)$ and $-2/(d+2)$. Keeping the duration fixed leaves finite-difference bias. These are checks of the proposed design on two fixed laws; the general lower bound, where applicable, is a mathematical claim rather than a conclusion from the fitted slopes.

\begin{figure}[!htbp]
\centering
\includegraphics[width=\linewidth]{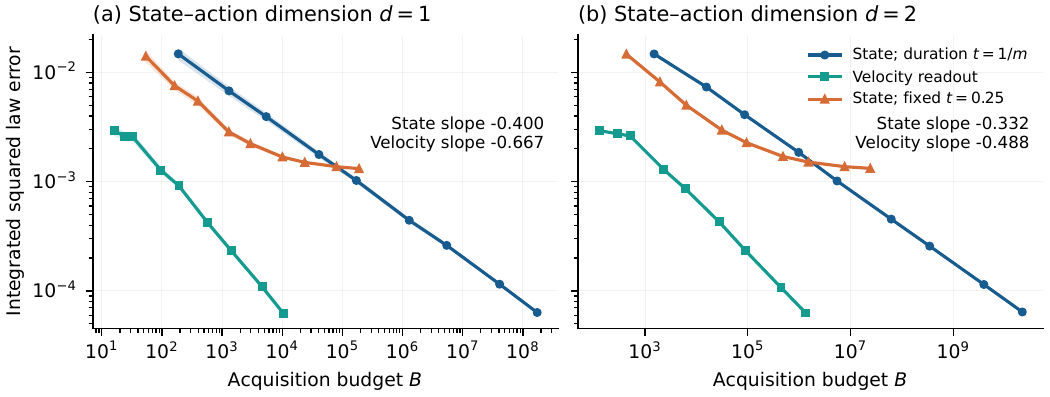}
\caption{Acquisition cost of the declared designs for state--action dimensions $d=1,2$. Costs include resets, readouts, and elapsed time. Shading is $\pm1.96$ standard errors over 64 noise replicates. Slopes use the five finest grids. The fixed-duration estimator retains differentiation bias. These curves illustrate estimator behavior on fixed laws; the minimax lower bound is proved separately.}
\label{fig:scaling}
\end{figure}

\section{Consequences and Scope}\label{sec:discussion}
The theory separates three questions. Rate--distortion describes the information required for a specified physical fidelity. Experimental completeness asks whether the apparatus can reveal the relevant distinctions. Inverse stability and acquisition cost determine how accurately they can be measured. A low prediction loss becomes evidence of physical recovery only when the target domain, observation model, and experimental access connect it to that claim.

Several consequences follow. Additional observations of an equivalent world pair cannot remove its law ambiguity. More expressive models cannot distinguish identical transcript distributions. Excitation and calibration solve different problems: excitation changes which dynamics are revealed, while calibration can anchor their physical coordinates. Even a correctly identified sampled transition map need not uniquely identify a continuous-time generator without suitable sampling assumptions \citep{zeng2024}. These are failures of particular implications, rather than claims that predictive world models are ineffective for their intended control tasks.

The running example gives these conclusions an operational meaning. An evaluation set containing only the observed circular motion rewards both laws equally. A region containing other radii asks a question the same data cannot answer. Releasing the object changes the evidence while keeping the physical target fixed. Moving the release point or readout time changes the strength and price of that evidence. The progression is from deciding which physical distinctions matter, to making them observable, to measuring them precisely. A larger training set helps at the last step only if its acquisition protocol supports the earlier ones.

Recovering a field also differs from discovering a preferred symbolic expression. The loss in~\eqref{eq:metric} compares physical derivatives over $D$. Two expressions defining the same field there have zero distance, and extensions differing only outside $D$ are not separated. The results support a specified domain of physical validity. They do not assert that a finite experiment uniquely selects a universal formula or an internal neural representation. Symbolic, numerical, and learned representations can thus be judged against the same physical question while retaining their distinct representational aims.

The response-codebook construction offers a way to understand what a learned predictor would have to accomplish. Its observations must resolve candidate fields at the requested physical accuracy; its representation must retain the distinctions needed by the decoder; and its approximation and optimization errors must fit within the remaining error allowance. The present theorems establish statistical possibility by constructing one decoder. Applying them to a particular encoder, latent transition model, and decoder requires an additional argument linking that implementation to physical-coordinate field error. The continuous-force example makes one part of this issue visible: repeating an informative experiment reduces measurement noise, whereas refining a fixed response codebook reduces its approximation error. These operations address different contributions to the reconstruction error.

Unknown sensing makes the physical-coordinate requirement especially consequential. In the gain example, an entire family of different physical gains produces the same observable dynamics. A known-state reference anchors the missing scale. If such references are unavailable, one may instead seek a model of observable dynamics or a law defined only up to the unresolved coordinate transformation. That is a legitimate alternative target, but its equivalence relation and loss must be stated accordingly. It cannot be converted into the fixed-coordinate recovery claim merely by attaching physical names to latent variables. The formulation makes the required references explicit rather than assigning them implicitly to the representation.

The results are statistical guarantees for a correctly specified compact world class, repeatable world-independent Gaussian noise, and the declared apparatus. The constructive codebook can be computationally prohibitive, and its sufficient budget need not match the information converse for every apparatus. Physical-coordinate recovery requires the necessary physical references; arbitrary latent coordinates and passive video are not covered by the sufficient access conditions. The nonparametric budget theorem measures the entire first-order state, fixes constant actions during each episode, and assumes strictly positive noise. Smoother classes, partial-state observations, model misspecification, and control waveforms within episodes require separate analyses. The simulations isolate these mechanisms in small systems and do not establish an advantage for any neural architecture.

The central claim is therefore operational: a world model can be certified as a recovered physical law only relative to a stated target and experiments capable of separating its competing explanations. This perspective makes the acquisition protocol part of the learning problem and exposes the cost of turning predictive evidence into physical knowledge.

\clearpage
\phantomsection
\addcontentsline{toc}{section}{References}
\bibliographystyle{plainnat}
\begingroup
\interlinepenalty=10000
\bibliography{references}

@article{shannon1948,
  author = {Shannon, Claude E.},
  title = {A Mathematical Theory of Communication},
  journal = {The Bell System Technical Journal},
  volume = {27},
  pages = {379--423, 623--656},
  year = {1948},
  doi = {10.1002/j.1538-7305.1948.tb01338.x},
  url = {https://people.math.harvard.edu/~ctm/home/text/others/shannon/entropy/entropy.pdf}
}

@inproceedings{shannon1959,
  author = {Shannon, Claude E.},
  title = {Coding Theorems for a Discrete Source with a Fidelity Criterion},
  booktitle = {IRE National Convention Record, Part 4},
  volume = {7},
  pages = {142--163},
  year = {1959},
  url = {https://doi.org/10.1109/9780470544242.ch21}
}

@book{cover2006,
  author = {Cover, Thomas M. and Thomas, Joy A.},
  title = {Elements of Information Theory},
  publisher = {Wiley-Interscience},
  edition = {2},
  year = {2006},
  doi = {10.1002/047174882X},
  url = {https://doi.org/10.1002/047174882X}
}

@article{pan2016,
  author = {Pan, Wei and Yuan, Ye and Gon{\c{c}}alves, Jorge and Stan, Guy-Bart},
  title = {A Sparse {Bayesian} Approach to the Identification of Nonlinear State-Space Systems},
  journal = {IEEE Transactions on Automatic Control},
  volume = {61},
  number = {1},
  pages = {182--187},
  year = {2016},
  doi = {10.1109/TAC.2015.2426291},
  url = {https://arxiv.org/abs/1408.3549}
}

@article{brunton2016,
  author = {Brunton, Steven L. and Proctor, Joshua L. and Kutz, J. Nathan},
  title = {Discovering Governing Equations from Data by Sparse Identification of Nonlinear Dynamical Systems},
  journal = {Proceedings of the National Academy of Sciences},
  volume = {113},
  number = {15},
  pages = {3932--3937},
  year = {2016},
  doi = {10.1073/pnas.1517384113},
  url = {https://arxiv.org/abs/1509.03580}
}

@article{bellman1970,
  author = {Bellman, Richard and {\AA}str{\"o}m, Karl Johan},
  title = {On Structural Identifiability},
  journal = {Mathematical Biosciences},
  volume = {7},
  number = {3--4},
  pages = {329--339},
  year = {1970},
  doi = {10.1016/0025-5564(70)90132-X},
  url = {https://www.sciencedirect.com/science/article/pii/002555647090132X}
}

@article{sontag2017,
  author = {Sontag, Eduardo D.},
  title = {Dynamic Compensation, Parameter Identifiability, and Equivariances},
  journal = {PLOS Computational Biology},
  volume = {13},
  number = {4},
  pages = {e1005447},
  year = {2017},
  doi = {10.1371/journal.pcbi.1005447},
  url = {https://journals.plos.org/ploscompbiol/article?id=10.1371/journal.pcbi.1005447}
}

@article{ha2018,
  author = {Ha, David and Schmidhuber, J{\"u}rgen},
  title = {World Models},
  journal = {arXiv preprint arXiv:1803.10122},
  year = {2018},
  url = {https://arxiv.org/abs/1803.10122}
}

@article{hafner2025,
  author = {Hafner, Danijar and Pasukonis, Jurgis and Ba, Jimmy and Lillicrap, Timothy},
  title = {Mastering Diverse Control Tasks through World Models},
  journal = {Nature},
  volume = {640},
  number = {8059},
  pages = {647--653},
  year = {2025},
  doi = {10.1038/s41586-025-08744-2},
  url = {https://www.nature.com/articles/s41586-025-08744-2}
}

@inproceedings{chen2018,
  author = {Chen, Ricky T. Q. and Rubanova, Yulia and Bettencourt, Jesse and Duvenaud, David K.},
  title = {Neural Ordinary Differential Equations},
  booktitle = {Advances in Neural Information Processing Systems},
  volume = {31},
  year = {2018},
  url = {https://papers.nips.cc/paper_files/paper/2018/hash/69386f6bb1dfed68692a24c8686939b9-Abstract.html}
}

@article{mania2022,
  author = {Mania, Horia and Jordan, Michael I. and Recht, Benjamin},
  title = {Active Learning for Nonlinear System Identification with Guarantees},
  journal = {Journal of Machine Learning Research},
  volume = {23},
  number = {32},
  pages = {1--30},
  year = {2022},
  url = {https://jmlr.org/papers/v23/20-807.html}
}

@article{zeng2024,
  author = {Zeng, Zhexuan and Yue, Zuogong and Mauroy, Alexandre and Gon{\c{c}}alves, Jorge and Yuan, Ye},
  title = {A Sampling Theorem for Exact Identification of Continuous-Time Nonlinear Dynamical Systems},
  journal = {IEEE Transactions on Automatic Control},
  volume = {69},
  number = {12},
  pages = {8402--8417},
  year = {2024},
  doi = {10.1109/TAC.2024.3409639},
  url = {https://arxiv.org/abs/2204.14021}
}

@article{donoho1991,
  author = {Donoho, David L. and Liu, Richard C.},
  title = {Geometrizing Rates of Convergence, {II}},
  journal = {The Annals of Statistics},
  volume = {19},
  number = {2},
  pages = {633--667},
  year = {1991},
  doi = {10.1214/aos/1176348114},
  url = {https://web.stanford.edu/dept/statistics/cgi-bin/donoho/wp-content/uploads/2018/08/GeoII.pdf}
}

@article{donoho1994,
  author = {Donoho, David L.},
  title = {Statistical Estimation and Optimal Recovery},
  journal = {The Annals of Statistics},
  volume = {22},
  number = {1},
  pages = {238--270},
  year = {1994},
  doi = {10.1214/aos/1176325367},
  url = {https://doi.org/10.1214/aos/1176325367}
}

@incollection{yu1997,
  author = {Yu, Bin},
  title = {{Assouad}, {Fano}, and {Le Cam}},
  booktitle = {Festschrift for Lucien Le Cam},
  editor = {Pollard, David and Torgersen, Erik and Yang, Grace L.},
  publisher = {Springer},
  pages = {423--435},
  year = {1997},
  doi = {10.1007/978-1-4612-1880-7_29},
  url = {https://link.springer.com/chapter/10.1007/978-1-4612-1880-7_29}
}

@article{chernoff1959,
  author = {Chernoff, Herman},
  title = {Sequential Design of Experiments},
  journal = {The Annals of Mathematical Statistics},
  volume = {30},
  number = {3},
  pages = {755--770},
  year = {1959},
  doi = {10.1214/aoms/1177706205},
  url = {https://doi.org/10.1214/aoms/1177706205}
}

@article{kaufmann2016,
  author = {Kaufmann, Emilie and Capp{\'e}, Olivier and Garivier, Aur{\'e}lien},
  title = {On the Complexity of Best-Arm Identification in Multi-Armed Bandit Models},
  journal = {Journal of Machine Learning Research},
  volume = {17},
  number = {1},
  pages = {1--42},
  year = {2016},
  url = {https://jmlr.org/papers/v17/kaufman16a.html}
}

@article{stone1982,
  author = {Stone, Charles J.},
  title = {Optimal Global Rates of Convergence for Nonparametric Regression},
  journal = {The Annals of Statistics},
  volume = {10},
  number = {4},
  pages = {1040--1053},
  year = {1982},
  doi = {10.1214/aos/1176345969},
  url = {https://doi.org/10.1214/aos/1176345969}
}

@article{yang1999,
  author = {Yang, Yuhong and Barron, Andrew},
  title = {Information-Theoretic Determination of Minimax Rates of Convergence},
  journal = {The Annals of Statistics},
  volume = {27},
  number = {5},
  pages = {1564--1599},
  year = {1999},
  doi = {10.1214/aos/1017939142},
  url = {https://projecteuclid.org/journals/annals-of-statistics/volume-27/issue-5/Information-theoretic-determination-of-minimax-rates-ofconvergence/10.1214/aos/1017939142.full}
}

@article{schoetz2026,
  author = {Sch{\"o}tz, Christof},
  title = {Nonparametric Estimation of Ordinary Differential Equations: Snake and Stubble},
  journal = {Bernoulli},
  volume = {32},
  number = {3},
  pages = {1880--1906},
  year = {2026},
  doi = {10.3150/25-BEJ1936},
  url = {https://arxiv.org/abs/2407.14989}
}

@article{schoetz_siebel2026,
  author = {Sch{\"o}tz, Christof and Siebel, Maximilian},
  title = {Lower Bounds for Nonparametric Estimation of Ordinary Differential Equations},
  journal = {Electronic Journal of Statistics},
  volume = {20},
  number = {1},
  pages = {503--559},
  year = {2026},
  doi = {10.1214/26-EJS2484},
  url = {https://arxiv.org/abs/2407.14993}
}

@inproceedings{hafner2019planet,
  author = {Hafner, Danijar and Lillicrap, Timothy and Fischer, Ian and Villegas, Ruben and Ha, David and Lee, Honglak and Davidson, James},
  title = {Learning Latent Dynamics for Planning from Pixels},
  booktitle = {Proceedings of the 36th International Conference on Machine Learning},
  series = {Proceedings of Machine Learning Research},
  volume = {97},
  pages = {2555--2565},
  year = {2019},
  publisher = {PMLR},
  url = {https://proceedings.mlr.press/v97/hafner19a.html}
}

@article{assran2023ijepa,
  author = {Assran, Mahmoud and Duval, Quentin and Misra, Ishan and Bojanowski, Piotr and Vincent, Pascal and Rabbat, Michael and LeCun, Yann and Ballas, Nicolas},
  title = {Self-Supervised Learning from Images with a Joint-Embedding Predictive Architecture},
  journal = {arXiv preprint arXiv:2301.08243},
  year = {2023},
  doi = {10.48550/arXiv.2301.08243},
  url = {https://arxiv.org/abs/2301.08243}
}

@article{bardes2024vjepa,
  author = {Bardes, Adrien and Garrido, Quentin and Ponce, Jean and Chen, Xinlei and Rabbat, Michael and LeCun, Yann and Assran, Mahmoud and Ballas, Nicolas},
  title = {Revisiting Feature Prediction for Learning Visual Representations from Video},
  journal = {arXiv preprint arXiv:2404.08471},
  year = {2024},
  doi = {10.48550/arXiv.2404.08471},
  url = {https://arxiv.org/abs/2404.08471}
}

@inproceedings{locatello2019,
  title = {Challenging Common Assumptions in the Unsupervised Learning of Disentangled Representations},
  author = {Locatello, Francesco and Bauer, Stefan and Lucic, Mario and Raetsch, Gunnar and Gelly, Sylvain and Sch{\"o}lkopf, Bernhard and Bachem, Olivier},
  booktitle = {Proceedings of the 36th International Conference on Machine Learning},
  pages = {4114--4124},
  year = {2019},
  volume = {97},
  series = {Proceedings of Machine Learning Research},
  publisher = {PMLR},
  url = {https://proceedings.mlr.press/v97/locatello19a.html}
}

@article{schoelkopf2021,
  author = {Sch{\"o}lkopf, Bernhard and Locatello, Francesco and Bauer, Stefan and Ke, Nan Rosemary and Kalchbrenner, Nal and Goyal, Anirudh and Bengio, Yoshua},
  title = {Towards Causal Representation Learning},
  journal = {arXiv preprint arXiv:2102.11107},
  year = {2021},
  url = {https://arxiv.org/abs/2102.11107}
}
\endgroup
\clearpage
\appendix
\section{Formal Experiment Model and Notation}\label{app:model}
The world class $\W$ is a nonempty compact metric space with its Borel $\sigma$-algebra. Its target map $w\mapsto F_w|_D$ is continuous into a complete separable metric space $(\F,\loss)$. Consequently the image $\F_0$ is compact and has finite diameter $D_F$. In the integrated-error formulation the ambient output space is $L^2(D,\nu;\R^p)$, while the admitted target fields are continuous. A continuous difference that is nonzero at a point is bounded away from zero on a relative neighborhood; full support of $\nu$ gives that neighborhood positive measure. Thus zero integrated distance identifies equal continuous fields on $D$.

The catalog $\A$ is a standard Borel space. The positive integer observation dimension $k_a$ and cost $c_a\in[c_{\min},\infty)$ are measurable. Variable-dimensional observations lie in the countable disjoint union of their Euclidean spaces. The Gaussian kernels in~\eqref{eq:episode} are jointly Borel in $(w,a)$, and the mean $m_a$ is continuous in $w$ for each fixed $a$. Every trajectory used by an episode exists uniquely and remains in $\mathcal X$ through its final readout. For an ordinary trajectory episode with $J_a$ readout times and $s$ sensor coordinates,
\[
 k_a=sJ_a,\qquad
 m_a(w)=\big(H_w(x_w^a(t_{a,1})),\ldots,H_w(x_w^a(t_{a,J_a}))\big),
\]
where $x_w^a$ is the solution with the preparation and input prescribed by $a$. A known-state calibration at reference $x_{\rm ref}$ instead has noiseless mean $H_w(x_{\rm ref})$. These means are generated by physical operations; they are not oracle evaluations of the vector field.

For the Gaussian calculations, let $\Sigma_a^{-1/2}$ be the inverse of the symmetric positive definite square root of $\Sigma_a$. Then $Z_a=\Sigma_a^{-1/2}Y_a$ has law $N(\mu_a(w),I_{k_a})$, where $\mu_a(w)=\Sigma_a^{-1/2}m_a(w)$. The transformation is known and invertible for every selected episode. Replacing raw readings by whitened readings in a transcript therefore preserves its information, likelihood ratios, and available decisions. It also preserves continuity and the equivalence relation~\eqref{eq:equivalence}.

A policy is a sequence of measurable kernels selecting an episode or a stop symbol from completed histories and internal randomness $R_\pi$, independent of the world and all observation innovations. The decoder $\decoder$ receives the selected operations and their observations, as well as the specified accessible part of $R_\pi$. Under budget $B$, a stopped history can be padded to length $\lfloor B/c_{\min}\rfloor$ with a deterministic null symbol. The null operation has no observation and contributes no information or likelihood ratio. This finite padding justifies the adaptive chain rules used below, including randomized stopping. Decoder randomization uses an independent seed. Hidden policy randomness may be adjoined to the transcript when deriving an upper bound on information.

No invariant meaning is assigned to a latent coordinate by predictive accuracy alone. If $S$ is a $C^1$ diffeomorphism on a physical state region, writing $z=S(x)$ transforms the dynamics and sensor to
\[
 \dot z=DS(S^{-1}z)F_w(S^{-1}z,u),\qquad
 \widetilde H(z)=H_w(S^{-1}z).
\]
Here $DS$ is the Jacobian of $S$. The observed trajectories agree when the permitted preparations transform consistently. Known physical preparations can break this equivalence; arbitrary coordinate changes do not preserve an apparatus with fixed references. The target in~\eqref{eq:metric} is explicitly in physical coordinates, so a quotient by coordinate transformations would be a different learning objective.

The table collects the main notation for reference; each quantity is introduced in the text before it is used in a result.
\begin{center}\small
\begin{tabular}{p{0.26\linewidth}p{0.66\linewidth}}
\toprule
Symbol & Meaning \\
\midrule
$x,u;\ p,q,s$ & Physical state and input; state, input, and sensor dimensions. \\
$w=(F_w,H_w),\ \W$ & Joint dynamics--sensor world and candidate world class. \\
$\mathcal X,K,U;\ D,\nu$ & State domain, evaluation state region, input set; evaluation domain $D=K\times U$ and probability measure. \\
$\F,\F_0,\loss$ & Ambient output space, candidate target set, and law metric. \\
$\lawhat,\widetilde F$ & Actual decoded field and abstract rate--distortion reconstruction. \\
$a,\A;\ k_a,c_a$ & Episode, apparatus catalog; output dimension and cost. \\
$m_a,\Sigma_a,\mu_a$ & Noiseless output mean, noise covariance, whitened mean. \\
$Y_a,Z_a$ & Raw and whitened experimental observations. \\
$\pi,\Pi_{\A}(B),T^\pi$ & Acquisition policy, budget-feasible policies, complete transcript. \\
$\decoder,R_\pi$ & Field decoder and independent policy randomness. \\
$\varepsilon,B,D_F$ & Squared-error tolerance, acquisition budget, target diameter. \\
$\mathfrak R^\star_{\A}(B),B^\star_{\A}(\varepsilon)$ & Minimax law risk and least feasible acquisition budget. \\
$\mathsf H_{\F_0}(r)$ & Metric entropy at physical resolution $r$. \\
$P,\rd_P,\capinfo_{\A}(P,B)$ & Joint prior, law rate--distortion, experimental information capacity. \\
$\kappa_{\A}$ & Largest physical distance within an experimental equivalence class. \\
$A,Q_A,\cost(A)$ & Finite program, its output dimension and one-round cost. \\
$\Delta_A(r),M_A(h)$ & Response gap at law resolution $r$; response-cover size at radius $h$. \\
$\omega_A(s),\mathfrak R_A(n)$ & Inverse response modulus; risk from $n$ rounds of fixed $A$. \\
$d=p+q,\ L,M,\sigma$ & State--action dimension, field Lipschitz constant, magnitude and noise scale in Section~\ref{sec:budget}. \\
\bottomrule
\end{tabular}
\end{center}

\section{Information Bounds and Experimental Ambiguity}\label{app:information}
All information quantities below use natural logarithms. The target is the random field $F=F_W$, while a prior on $W$ specifies both dynamics and sensing. Regular conditional distributions exist on the standard Borel spaces of Appendix~\ref{app:model}. Thus an actual randomized decoder induces an admissible abstract channel in~\eqref{eq:rd}; the data-processing argument in Theorem~\ref{thm:rd} applies without assuming that sensing is known. The supremum over priors in~\eqref{eq:rdbudget} is taken after the necessary budget bound for each prior. No interchange of a policy infimum and a prior supremum is used, and no minimax-optimal procedure is assumed to attain its infimum.

We first prove the metric-entropy statements. A finite radius-$r$ cover with centers $f_1,\ldots,f_N\in\F_0$ admits a measurable nearest-center map, with ties resolved by the smallest index. Its reconstruction has squared loss at most $r^2$ pointwise and takes at most $N$ values. Consequently its mutual information with $F$ is at most $\log N$, proving $\rd_P(r^2)\le\mathsf H_{\F_0}(r)$ for every prior. This is a source-description statement; its encoder is supplied with $F$.

\begin{proof}[Proof of Corollary~\ref{cor:packing}]
Take target fields $f_1,\ldots,f_M$ with pairwise distances strictly greater than $2r$, and choose a world $w_j$ representing each field. Let $J$ be uniform on $\{1,\ldots,M\}$ and set $W=w_J$, $F=f_J$. For any abstract reconstruction $\widetilde F$ with $\E\loss(F,\widetilde F)^2\le\varepsilon$, decode
\[
 \widehat J=\min\argmin_{1\le j\le M}\loss(f_j,\widetilde F).
\]
If $\widehat J\ne J$, the triangle inequality and nearest-center rule imply
\[
 2r<\loss(f_J,f_{\widehat J})
 \le\loss(f_J,\widetilde F)+\loss(\widetilde F,f_{\widehat J})
 \le2\loss(f_J,\widetilde F).
\]
Thus its error probability $p_e$ is at most $\varepsilon/r^2$ by Markov's inequality. Since $J$ and $F$ determine each other on this packing, Fano's inequality yields
\begin{align}
 I(F;\widetilde F)
 &=I(J;\widetilde F)\nonumber\\
 &\ge\log M-h_2(p_e)-p_e\log(M-1)\nonumber\\
 &\ge(1-\varepsilon/r^2)\log M-\log2,\label{eq:app_fano}
\end{align}
where $h_2(t)=-t\log t-(1-t)\log(1-t)$, with $0\log0=0$. Taking the infimum over abstract channels proves the rate--distortion lower bound. Theorem~\ref{thm:rd} then gives~\eqref{eq:packing} for any uniformly $\varepsilon$-accurate procedure. Compactness bounds the packing size at each positive resolution. Taking $r=2\sqrt\varepsilon$ gives $1-\varepsilon/r^2=3/4$ and $2r=4\sqrt\varepsilon$, which is the final assertion.
\end{proof}
This lower bound uses a selected packing prior; it does not identify the rate--distortion function of every prior with metric entropy.

For completeness, we give a two-world bound that allows arbitrary measurable outputs, including fields outside $\F_0$. For a fixed policy let $P_w^\pi$ and $P_v^\pi$ be its transcript laws in worlds $w$ and $v$, and put $\delta=\loss(F_w,F_v)$. Let $\lambda$ dominate these two distributions, with respective densities $p_w,p_v$, and let $K(d\widehat f\mid t)$ be the decoder's output kernel. The triangle inequality gives
\[
 \loss(F_w,\widehat f)^2+\loss(F_v,\widehat f)^2\ge\delta^2/2.
\]
Integrating this inequality against $K(d\widehat f\mid t)\min(p_w(t),p_v(t))\,d\lambda(t)$ bounds the sum of the two risks from below by $\delta^2(1-\TV(P_w^\pi,P_v^\pi))/2$. The larger risk is at least half that sum, so
\begin{equation}\label{eq:app_twopoint}
 \max_{s\in\{w,v\}}\E_s\loss(F_s,\lawhat)^2
 \ge\frac{\loss(F_w,F_v)^2}{4}
          \bigl(1-\TV(P_w^\pi,P_v^\pi)\bigr).
\end{equation}
Here $\TV(P,Q)=\sup_E|P(E)-Q(E)|$. The conclusion also holds if a risk is infinite.

\begin{proof}[Proof of Proposition~\ref{prop:floor}]
Each equality constraint $\mu_a(w)=\mu_a(v)$ defines a closed subset of compact $\W^2$. Their intersection is nonempty, since it contains the diagonal, and the continuous distance $\loss(F_w,F_v)$ attains its maximum $\kappa_\A$ on this intersection. If $w\sim_\A v$, pad the transcript after stopping to $N_B=\lfloor B/c_{\min}\rfloor$ steps with a null operation and constant output. At a shared history the policy chooses the next operation by the same world-independent rule; a nonnull operation has the same Gaussian response distribution in the two worlds. Induction over the finite padded transcript proves $P_w^\pi=P_v^\pi$. Independent policy randomization can be conditioned on throughout this argument. Applying~\eqref{eq:app_twopoint} to an attaining pair and taking the infimum over policies and decoders proves~\eqref{eq:floor}, including $B=0$.
\end{proof}

A quantitative extension measures the best squared response separation per unit cost:
\begin{equation}\label{eq:app_gamma}
 \Gamma_\A(w,v)=\sup_{a\in\A}
       \frac{\|\mu_a(w)-\mu_a(v)\|_2^2}{c_a}\in[0,\infty].
\end{equation}
For $B>0$ and $\Gamma_\A(w,v)<\infty$, the adaptive divergence chain rule gives
\begin{equation}\label{eq:app_adaptive_kl}
 \KL(P_w^\pi\|P_v^\pi)
 \le\frac12\E_w\sum_{j=1}^{N_\pi}
             \|\mu_{a_j}(w)-\mu_{a_j}(v)\|_2^2
 \le\frac B2\Gamma_\A(w,v).
\end{equation}
To justify the first inequality, augment the record with the policy's entire independent random seed and pad after stopping as above. The seed has the same distribution in both worlds. Conditional on the seed and completed history, the selected operation is fixed and contributes no likelihood ratio. Its response contributes
$\KL(N(\mu_a(w),I)\|N(\mu_a(v),I))=\|\mu_a(w)-\mu_a(v)\|_2^2/2$.
The chain rule gives equality for this augmented record; discarding any inaccessible seed variables cannot increase divergence. The second inequality follows from the pathwise cost bound. Pinsker's inequality and~\eqref{eq:app_twopoint} therefore give
\begin{equation}\label{eq:app_cost_lower}
 \mathfrak R_\A^\star(B)
 \ge\frac{\loss(F_w,F_v)^2}{4}
 \left(1-\frac{\sqrt{B\Gamma_\A(w,v)}}2\right)_+,
 \qquad (t)_+=\max\{t,0\}.
\end{equation}
When $\Gamma_\A(w,v)=\infty$, the asserted nonnegative lower bound is simply zero. This argument applies to adaptation between complete episodes, as declared in Section~\ref{sec:problem}; it does not silently enlarge the catalog to include feedback inside an episode.

The scalar information converse has no unconditional sufficiency direction. An explicit example uses two-dimensional constant fields
\[
 F_{s,v}(x)=(\alpha s,v),\qquad
 s\in\{-1,-1/3,1/3,1\},\quad v\in\{-1,1\},\quad0<\alpha<1.
\]
Use the known sensor $H(x)=x_1$. Every permitted episode resets $x(0)=0$, runs for unit time, and returns $Y=\alpha S+G$, where $G\sim N(0,1)$; repetitions have independent noise and unit cost. Evaluate squared field error on any compact physical domain with a full-support probability measure $\nu$. Let the prior make $S$ and $V$ independent and uniform on their respective sets. There are only eight worlds, so all compactness and continuity assumptions hold.

Every transcript, including adaptive stopping and independent randomization, is independent of $V$. For any output field with second component $\widehat f_2$, if its risk is finite, independence and $\E V=0$ yield
\[
 \E\|V\mathbf1-\widehat f_2\|_{L^2(\nu)}^2
 =1+\E\|\widehat f_2\|_{L^2(\nu)}^2\ge1.
\]
Here $\mathbf1$ is the constant-one function and $\nu$ has total mass one. Infinite-risk outputs obey the same lower bound trivially. Thus the prior-average risk of every executable procedure, and hence its worst-world risk, is at least one at every budget.

An abstract reconstruction can instead report $V$ and output the constant field $\widetilde F=(0,V)$. Since $\E S^2=5/9$,
\[
 \E\loss(F,\widetilde F)^2=5\alpha^2/9\le\alpha^2,
 \qquad I(F;\widetilde F)=H(V)=\log2,
\]
so $\rd_P(\alpha^2)\le\log2$. This reconstruction is allowed to lie outside the candidate class, as required by the output convention of Section~\ref{sec:problem}. On the other hand, $n$ permitted readings have mean $\overline Y=\alpha S+N(0,1/n)$. The four means have minimum spacing $2\alpha/3$, so nearest-mean decoding of $S$ has error probability at most $2\Phi(-\alpha\sqrt n/3)$, where $\Phi$ is the standard normal distribution function. This probability tends to zero. Fano's inequality implies
\[
 I(F;Y_1,\ldots,Y_n)=I(S;Y_1,\ldots,Y_n)\longrightarrow\log4.
\]
For some finite budget, therefore,
\[
 \capinfo_\A(P,B)>\log2\ge\rd_P(\alpha^2),
 \qquad
 \inf_{\pi,\decoder}\E_P\loss(F,\decoder(T^\pi))^2\ge1>\alpha^2.
\]
The counterexample concerns which part of the law is observed, rather than a failure of rate--distortion theory. The abstract source encoder is allowed to reveal $V$; the declared sensor is not. The target information needed for a useful reconstruction cannot in general be replaced by an equal amount of information about different physical distinctions.

We finally bound the information supplied per unit experimental cost. Define
\begin{equation}\label{eq:app_chi}
 \chi_\A=\sup_{a\in\A}\sup_{Q\in\mathcal P(\W)}
             \frac{I_Q(F;Z_a)}{c_a}\in[0,\infty],
\end{equation}
where $\mathcal P(\W)$ denotes all Borel priors on the joint world space, $W\sim Q$, and $F=F_W$. This supremum must include every such prior, since conditional distributions after adaptive observations need not remain in a smaller prior family. For every prior $P$ and $B>0$,
\begin{equation}\label{eq:app_capacity_cost}
 \capinfo_\A(P,B)\le B\chi_\A.
\end{equation}
At zero budget capacity is zero; if $\chi_\A<\infty$,~\eqref{eq:app_capacity_cost} holds there as well. A nontrivial finite bound requires a finite $\chi_\A$ and is not a consequence of compactness alone.

To prove~\eqref{eq:app_capacity_cost}, let $R_\pi$ contain all policy randomization, independent of both $W$ and the observation innovations. It may be added to the transcript even when only part is accessible. Pad to $N_B$ steps and put $\mathscr H_0=\sigma(R_\pi)$ and
$\mathscr H_j=\sigma(R_\pi,a_1,Z_1,\ldots,a_j,Z_j)$.
Conditional on $\mathscr H_{j-1}$ the chosen episode or null action is fixed. Since $I_P(F;R_\pi)=0$, the information chain rule gives
\[
 I_P(F;T^\pi)
 \le\sum_{j=1}^{N_B}I_P(F;Z_j\mid\mathscr H_{j-1},a_j).
\]
At almost every completed history, the conditional law of $W$ is a Borel prior $Q_j$ on the joint world space. For a nonnull selected episode $a$, its conditional response is the original episode channel under $Q_j$: independence of innovations ensures that conditioning on the policy seed does not reveal its future measurement noise. Its conditional mutual information is thus at most $c_a\chi_\A$ by~\eqref{eq:app_chi}; a null episode contributes zero. If $\chi_\A<\infty$, averaging over histories yields
\[
 I_P(F;T^\pi)
 \le\chi_\A\E_P\sum_{j=1}^{N_\pi}c_{a_j}\le B\chi_\A.
\]
Taking the supremum over admissible policies proves the result. If $\chi_\A=\infty$ and $B>0$, the bound is true in the extended-real sense without performing an indeterminate multiplication at zero.

In the whitened Gaussian model, write $D_a=\diam\mu_a(\W)$, which is finite for each fixed episode. For any prior $Q$, set $\overline\mu_a=\E_Q\mu_a(W)$. Comparing the conditional Gaussian response with $N(\overline\mu_a,I)$ gives
\begin{align}
 I_Q(F;Z_a)&\le I_Q(W;Z_a)\nonumber\\
 &\le\E_Q\KL\bigl(N(\mu_a(W),I)\|N(\overline\mu_a,I)\bigr)\nonumber\\
 &=\tfrac12\E_Q\|\mu_a(W)-\overline\mu_a\|_2^2
 \le\tfrac14D_a^2.\label{eq:app_gaussian_capacity}
\end{align}
The comparison inequality follows from the identity that the average conditional divergence to a reference distribution equals mutual information plus the divergence of the marginal response to that reference. The final step uses an independent copy $W'$: the variance is $\E\|\mu_a(W)-\mu_a(W')\|_2^2/2\le D_a^2/2$. Hence
\begin{equation}\label{eq:app_chi_diameter}
 \chi_\A\le\frac14\sup_{a\in\A}\frac{D_a^2}{c_a}.
\end{equation}
This supplies a finite information-per-cost bound when the displayed supremum is finite. It also shows why an episode catalog with arbitrarily large response diameter at bounded cost requires separate control. Neither a finite per-episode response dimension nor compactness of the world class imposes that control uniformly over the catalog.

\section{Constructive Recovery and Inverse Stability}\label{app:recovery}

Throughout this appendix, $\W$ is a nonempty compact metric space, the target map $w\mapsto F_w$ takes values continuously in the complete separable metric space $(\F,\loss)$, and each whitened response $\mu_a$ is continuous. The target diameter is $D_F>0$. For a finite nonempty program $A\subset\A$, write
\[
 \mu_A(w)=(\mu_a(w))_{a\in A},\qquad
 Q_A=\sum_{a\in A}k_a,\qquad
 \cost(A)=\sum_{a\in A}c_a.
\]
The order used in the stacked vector is fixed in advance. Independent repetitions give the averaged response~\eqref{eq:average}, $\overline Z_A=\mu_A(w)+G/\sqrt n$, where $G\sim N(0,I_{Q_A})$, at deterministic cost $n\cost(A)$. Response norms in this appendix are Euclidean and already account for the known observation covariance through whitening. The representatives in the response codebook are complete worlds, including sensor uncertainty; only their physical fields are returned.

\begin{proof}[Proof of Theorem~\ref{thm:recovery}]
Fix a true world $w\in\W$, and abbreviate $\mu=\mu_A(w)$ and $\mu_j=\mu_A(v_j)$. Let $M=M_A(h)$ be the size of the fixed codebook in~\eqref{eq:codebook}. Choose an index $j_0$ whose response satisfies $s_0:=\|\mu_{j_0}-\mu\|\le h$. This index is used only in the analysis and need not be known to the learner. If a candidate $j$ has $\loss(F_w,F_{v_j})\ge r$, the separation definition~\eqref{eq:sep} gives
\[
 s_j:=\|\mu_j-\mu\|\ge\Delta_A(r)>h\ge s_0.
\]
In particular, $\mu_j-\mu_{j_0}\ne0$. If this candidate is selected by~\eqref{eq:decoder}, its squared residual is at most the comparator's. Substitution of $\overline Z_A=\mu+G/\sqrt n$ and expansion give
\[
 2\langle G,\mu_j-\mu_{j_0}\rangle
 \ge\sqrt n(s_j^2-s_0^2).
\]
The left inner product divided by $\|\mu_j-\mu_{j_0}\|$ is standard normal. The triangle inequality gives
\[
 \frac{\sqrt n(s_j^2-s_0^2)}{2\|\mu_j-\mu_{j_0}\|}
 \ge\frac{\sqrt n(s_j-s_0)}2
 \ge\frac{\sqrt n(\Delta_A(r)-h)}2.
\]
The Gaussian tail bound therefore implies
\[
 \mathbb P_w\{\widehat j=j\}
 \le\exp\!\left[-\frac{n(\Delta_A(r)-h)^2}{8}\right].
\]
A union bound over the at most $M$ candidates whose target fields are at distance at least $r$ proves~\eqref{eq:decode_error}:
\begin{equation}\label{eq:appendix_decode_tail}
 \sup_{w\in\W}\mathbb P_w\{\loss(F_w,\lawhat)\ge r\}
 \le\min\!\left\{1,M\exp\!\left[-\frac{n(\Delta_A(r)-h)^2}{8}\right]\right\}.
\end{equation}
Both the target and the selected field belong to $\F_0$, so their distance is at most $D_F$. Splitting the expectation according to whether the distance is below $r$ gives~\eqref{eq:decode_risk}:
\[
 \sup_{w\in\W}\E_w\loss(F_w,\lawhat)^2
 \le r^2+D_F^2M\exp\!\left[-\frac{n(\Delta_A(r)-h)^2}{8}\right].
\]
All bounds are uniform in the true world, including worlds outside the codebook. Taking $r=\sqrt{\varepsilon/2}$, $h=\Delta_A(r)/2$, and
\[
 n=\left\lceil\frac{32}{\Delta_A(r)^2}
             \log\frac{2M_A(h)D_F^2}{\varepsilon}\right\rceil
\]
makes the two terms at most $\varepsilon/2$ each and proves~\eqref{eq:constructive_budget}. Here $\log$ is natural. Compactness of the continuous image $\mu_A(\W)$ supplies a finite response cover for every $h>0$. The smallest minimizing codebook index is a measurable function of the response vector, so the construction defines an admissible decoder.
\end{proof}

The same argument also gives a useful statement without choosing a physical resolution first. For any $t>0$, a candidate with response distance $s_j>h+t$ can beat the comparator only with probability at most $e^{-nt^2/8}$. Consequently, for $0<\eta<1$, the inverse modulus in~\eqref{eq:modulus} gives
\begin{equation}\label{eq:appendix_codebook_modulus}
 \sup_{w\in\W}\mathbb P_w\!\left\{
 \loss(F_w,\lawhat)>
 \omega_A\!\left(h+\sqrt{\frac{8\log(M_A(h)/\eta)}n}\right)
 \right\}\le\eta.
\end{equation}
Thus approximation of the response image, statistical precision, and conversion from response error to physical error enter separately. This is a statistical existence guarantee; it does not bound the computation required to construct or search the cover.

\begin{proof}[Proof of Corollary~\ref{cor:characterization}]
By~\eqref{eq:equivalence}, $\kappa_{\A}=0$ exactly when every pair of worlds with distinct target fields is separated by some episode. For $0<r\le D_F$, let
\[
 K_r=\{(w,v)\in\W^2:\loss(F_w,F_v)\ge r\}.
\]
This set is nonempty and compact. Suppose every pair with distinct target fields is separated by some admissible episode. The relatively open sets
\[
 U_a=\{(w,v)\in K_r:\|\mu_a(w)-\mu_a(v)\|>0\},
 \qquad a\in\A,
\]
cover $K_r$. A finite subcover defines a finite program $A_r$. Its stacked response distance is continuous and strictly positive on $K_r$, so it attains a positive minimum: $\Delta_{A_r}(r)>0$. Conversely, such a program for every $r$ separates any distinct target pair by taking $r=\loss(F_w,F_v)$. Theorem~\ref{thm:recovery} then provides, for every positive squared-error tolerance, a finite deterministic budget with uniformly bounded risk at that tolerance. Since $\mathfrak R_{\A}^\star(B)$ is nonincreasing in $B$, its limit is zero.

If pairwise separation fails, take worlds $w,v$ with $\delta=\loss(F_w,F_v)>0$ and identical responses under every episode. Their episode distributions agree. At a common history, an admissible policy uses the same next-episode and stopping kernels in the two worlds. Pad stopped transcripts with a null symbol to length $\lfloor B/c_{\min}\rfloor$; induction then gives identical complete transcript laws, including accessible independent policy randomness. For every possible output $f\in\F$,
\[
 \loss(F_w,f)^2+\loss(F_v,f)^2\ge\delta^2/2.
\]
Integration under the common transcript law and the common decoder kernel shows that the larger of the two risks is at least $\delta^2/4$ at every budget. Uniform consistency is therefore equivalent to pairwise separation and to positive finite-program separation at every positive resolution.

For completeness, consider the information-saturation formulation. Let $P$ be any finitely supported prior on worlds, let $W\sim P$, and let $F=F_W$. The entropy $\mathsf H_P(F)$ is the entropy of the finite set of distinct target fields under their induced probabilities, rather than the entropy of world labels. Recall that
\[
 \capinfo_{\A}(P,B)=\sup_{\pi\in\Pi_\A(B)}I_P(F;T^\pi).
\]
If the target is constant under $P$, both quantities vanish. Otherwise list the distinct target fields as $f_1,\ldots,f_m$, $m\ge2$, and put $\delta_P=\min_{i\ne j}\loss(f_i,f_j)>0$. Pairwise separation and the preceding construction produce a sequence of admissible procedures whose mean squared error under $P$ tends to zero as their budgets tend to infinity. Map each reconstructed field to the nearest member of this finite target list, breaking ties by the smallest index. If the decoded label is wrong, the original field error is at least $\delta_P/2$. Markov's inequality therefore bounds the label error $p_B$ by $4\E_P\loss(F,\lawhat)^2/\delta_P^2$, which tends to zero along the constructed sequence.

Fano's inequality gives
\[
 \mathsf H_P(F\mid T^\pi)\le h_2(p_B)+p_B\log(m-1)\longrightarrow0,
\]
where $h_2(p)=-p\log p-(1-p)\log(1-p)$ with $0\log0=0$. Independent decoder randomization can be included in the record without changing $I_P(F;T^\pi)$. It follows that $\capinfo_{\A}(P,B)\to \mathsf H_P(F)$: the lower limit follows from these procedures and monotonicity in budget, and $I_P(F;T^\pi)\le \mathsf H_P(F)$ gives the upper bound. All entropies here use natural logarithms; the assertion is unchanged if all information quantities use bits.

Conversely, if distinct target worlds $w,v$ have identical episode responses, take $P$ uniform on these two worlds. Every admissible transcript is independent of the binary target, so $\capinfo_{\A}(P,B)=0$ for every $B$, whereas $\mathsf H_P(F)=\log2>0$. Information saturation for every finitely supported prior therefore implies pairwise separation. This proves all formulations in the corollary. The universal quantifier over priors is essential: saturation for one prior need not detect ambiguity elsewhere in $\W$.
\end{proof}

The resolution quantifier cannot in general be interchanged with the choice of program. For example, take $\dot x=f(u)$, $x(0)=0$, known identity sensing, and the compact class of $1$-Lipschitz functions $f:[0,1]\to[-1,1]$. An episode holds one constant input for unit time and reads the terminal state with Gaussian noise. All such episodes together separate distinct functions. Any finite program, however, observes only finitely many input values and misses a sufficiently small Lipschitz bump between them. The characterization asserts $\forall r>0\;\exists A_r$, not the existence of one finite program valid at every resolution.

We next justify the fixed-program bounds~\eqref{eq:fixedmodulus}, including their measurability and limiting assertions. Denote by $\mathfrak R_A(n)$ the minimax squared target error over all estimators of the complete record from $n$ independent executions of the fixed program $A$. The constraint set in~\eqref{eq:modulus} is a nonempty compact subset of $\W^2$, so $\omega_A(s)$ is attained and lies in $[0,D_F]$. It is nondecreasing. It is also right-continuous: for $s_j\downarrow s$, choose maximizing pairs at $s_j$; compactness supplies a convergent subsequence, whose limiting response distance is at most $s$ and whose target distance equals $\lim_j\omega_A(s_j)$. Hence this limit is at most $\omega_A(s)$, and monotonicity gives equality. In particular $\omega_A$ is Borel measurable and $\omega_A(s)\downarrow\omega_A(0)$ as $s\downarrow0$.

An exact nearest-response estimator over $\W$ exists measurably. To see this without imposing a coordinate parameterization of $\W$, write
\[
 \ell(z,v)=\|z-\mu_A(v)\|,\qquad
 m(z)=\min_{v\in\W}\ell(z,v).
\]
The minimum exists, and $m$ is $1$-Lipschitz in $z$. The minimum of $\ell(z,\cdot)$ over any fixed nonempty compact subset of $\W$ is also $1$-Lipschitz. At each integer level $j\ge1$, choose a finite cover of $\W$ by closed metric balls of radius $2^{-j}$ with centers in $\W$. Select the first ball whose intersection with the previously selected balls contains a minimizer of $\ell(z,\cdot)$. Each test compares $m(z)$ with the minimum over a fixed compact intersection, taking the latter minimum as $+\infty$ for an empty intersection; thus each finite-level choice is Borel. The chosen intersections with the minimizer set are nested, nonempty, compact, and have diameter at most $2^{1-j}$. They determine a unique point $\widehat w(z)$, which is the limit of the selected centers and hence is Borel measurable. It is an exact minimizer. Composition with the continuous target map gives a measurable field estimator.

\begin{proof}[Proof of Proposition~\ref{prop:modulus}]
Use the exact nearest-response estimator above with $z=\overline Z_A$. The true world is a feasible comparator, and the triangle inequality gives
\[
 \|\mu_A(\widehat w)-\mu_A(w)\|
 \le 2\|G\|/\sqrt n.
\]
By the inverse modulus definition, the output $F_{\widehat w}$ therefore has risk bounded uniformly by
\[
 \mathfrak R_A(n)\le
 \E\left[\omega_A(2\|G\|/\sqrt n)^2\right],
 \qquad G\sim N(0,I_{Q_A}).
\]
For the lower bound, fix $t\ge0$ and choose worlds attaining $\omega_A(t)$, with actual response separation $s\le t$. Their full $n$-round Gaussian records have whitened mean separation $\sqrt n\,s$. If $P_w^{(n)}$ and $P_v^{(n)}$ denote these record laws and $\Phi$ is the standard normal cumulative distribution function, projection onto the line joining the means gives
\[
 1-\TV(P_w^{(n)},P_v^{(n)})
 =2\Phi(-\sqrt n\,s/2)
 \ge2\Phi(-\sqrt n\,t/2).
\]
For any output $f$, the sum of the two squared losses is at least $\omega_A(t)^2/2$. Integrating this inequality against the pointwise minimum of the two record densities, and also against the decoder kernel if randomized, bounds the average pairwise risk below by
\[
 \frac{\omega_A(t)^2}{4}
       \bigl(1-\TV(P_w^{(n)},P_v^{(n)})\bigr)
 \ge\frac{\omega_A(t)^2}{2}\Phi(-\sqrt n\,t/2).
\]
The maximal risk is no smaller than the average. Taking $t=n^{-1/2}$ proves
\[
 \frac{\Phi(-1/2)}2\,\omega_A(n^{-1/2})^2
 \le\mathfrak R_A(n)
 \le\E\left[\omega_A(2\|G\|/\sqrt n)^2\right].
\]
Taking $t=0$ also gives $\mathfrak R_A(n)\ge\omega_A(0)^2/4$. If $\omega_A(0)=0$, right-continuity at zero and bounded convergence in the upper bound imply $\mathfrak R_A(n)\to0$. If $\omega_A(0)>0$, the lower bound prevents consistency.

Finally suppose $cs^\beta\le\omega_A(s)\le Cs^\beta$ for $0\le s\le s_0$, with $0<\beta\le1$. For $n\ge s_0^{-2}$, the displayed two-point bound gives $\mathfrak R_A(n)\ge\tfrac12\Phi(-\tfrac12)c^2n^{-\beta}$. For the upper bound, split according to whether $2\|G\|/\sqrt n\le s_0$ and use $\omega_A\le D_F$. Then
\[
 \mathfrak R_A(n)
 \le C^2(4/n)^\beta\E\|G\|^{2\beta}
   +D_F^2\mathbb P\{\|G\|>s_0\sqrt n/2\}.
\]
The moment is finite. Since $\E e^{\|G\|^2/4}=2^{Q_A/2}$, the tail is at most $2^{Q_A/2}e^{-s_0^2n/16}$, which is at most a constant multiple of $n^{-\beta}$. This proves the matching upper order. The program $A$, its dimension $Q_A$, and the modulus constants are fixed in this rate statement.
\end{proof}

Identifiability does not itself imply a polynomial rate. In the scalar compact inverse problem with parameter $\theta\in[0,1]$, target distance $|\theta-\vartheta|$, and response $\mu(\theta)=e^{-1/\theta^2}$ for $\theta>0$, $\mu(0)=0$, the response is smooth and injective. Comparing $0$ with $\theta=[\log(1/s)]^{-1/2}$ shows $\omega(s)\ge[\log(1/s)]^{-1/2}$ for $0<s\le e^{-1}$. The fixed-program lower bound is therefore of order at least $1/\log n$, excluding every polynomial upper rate on this class.

The abstract separation condition has a direct sufficient realization for nonlinear physical dynamics. The next proposition spells out the required access rather than assuming that a learner can prescribe an arbitrary physical state through its latent representation.

\begin{proposition}[Sufficient physical access for nonlinear dynamics]\label{prop:physical_access}
Let $D=K\times U$, with $K\subset\R^p$ and $U\subset\R^q$ compact, and let $\W$ satisfy Assumption~\ref{ass:standing}(i),(ii),(v). Let $E\subset\R^p$ be a compact state neighborhood with $\delta=\operatorname{dist}(K,\R^p\setminus\operatorname{int}E)>0$. Suppose there are constants $M_0,L_0\ge0$ such that, for all $w\in\W$, $u\in U$, and $x,y\in E$, $\|F_w(x,u)\|\le M_0$ and $\|F_w(x,u)-F_w(y,u)\|\le L_0\|x-y\|$. Suppose the sensor maps are continuous and injective on $E$, and that both fields and sensors depend continuously on $w$ in the uniform norms on $E\times U$ and $E$, respectively. Assume finite-cost episodes include calibration at every known state in $E$, and evolution episodes that reset exactly to any state in $K$, hold any constant input in $U$, and read the state through the sensor once at a time $t$ drawn from a set of positive times with infimum zero. Fix $0<T_0<\delta/(M_0+1)$ and let $\A'$ consist of these calibrations and of these evolution episodes with $t\le T_0$. Then:
\begin{enumerate}[label=(\alph*)]
 \item every evolution episode in $\A'$ has a unique trajectory that remains in $E$, and every episode in $\A'$ has a continuous response map $w\mapsto m_a(w)$;
 \item for every pair $w,v\in\W$ with $F_w\ne F_v$ at some point of $D$, some episode $a\in\A'$ satisfies $m_a(w)\ne m_a(v)$.
\end{enumerate}
Consequently $\A'$ satisfies Assumption~\ref{ass:standing}(iii),(iv) and $\kappa_{\A'}=0$, so Corollary~\ref{cor:characterization} gives $\mathfrak R^\star_{\A'}(B)\to0$; since $\A'\subseteq\A$, also $\mathfrak R^\star_{\A}(B)\to0$ as $B\to\infty$.
\end{proposition}

\begin{proof}
Let $M_0$ be a common field bound and put $\delta=\operatorname{dist}(K,\R^p\setminus\operatorname{int}E)>0$. With $T_0$ as in the statement, every trajectory reset in $K$ under a constant input stays in $E$ through time $T_0$. Existence and uniqueness follow from the stated state-Lipschitz assumption. If $w_j\to w$, uniform convergence of the fields and Gronwall's inequality give uniform convergence of the associated trajectories on this interval. Uniform convergence of the sensors and uniform continuity of the limiting sensor on $E$ then give convergence of the noiseless readouts. Calibration means are continuous directly, and the fixed nonsingular whitening maps preserve continuity.

Take worlds $w,v$ with distinct target fields. If their sensors differ at any state of $E$, calibration at that known state separates their response means. Otherwise they share a common injective sensor $H$ on $E$. By hypothesis there is a point $(x_0,u_0)\in K\times U$ with $F_w(x_0,u_0)\ne F_v(x_0,u_0)$. Reset to $x_0$ and hold $u_0$. Since $\|x_w(s)-x_0\|\le M_0s$ and $F_w(\cdot,u_0)$ is $L_0$-Lipschitz on $E$, $\|x_w(t)-x_0-tF_w(x_0,u_0)\|\le L_0M_0t^2/2$ for $0\le t\le T_0$, and likewise for $v$. Hence
\[
 x_w(t)-x_v(t)
 =t\bigl(F_w(x_0,u_0)-F_v(x_0,u_0)\bigr)+O(t^2)
 \quad\text{as }t\downarrow0.
\]
The two states consequently differ for every sufficiently small positive time. They lie in $E$, so injectivity of $H$ makes their noiseless readings different. Some permitted positive readout time is small enough, and whitening preserves this difference. Thus an admissible episode separates the pair, proving the claim. Enlarging the catalog cannot increase $\mathfrak R^\star$, which gives the final assertion.
\end{proof}

The access assumptions are sufficient, not necessary. They require physical references and preparations supplied by the apparatus; they do not follow from model expressivity or accurate trajectory prediction. Without injective known sensing or adequate calibration, a change of physical coordinates can be compensated by a changed sensor. Without enough resets or inputs, laws can agree on all visited trajectories. Even when the proposition applies, its short-time argument gives separation rather than a quantitative polynomial inverse modulus; the budget needed to overcome a weak separation gap must still be analyzed.


\section{Proof of the Lipschitz budget theorem}
\label{app:lipschitz-budget}

We prove both bounds for the exact class, target loss, and observation
cost specified before Theorem~\ref{thm:lipschitz-budget}.
Throughout, $\varepsilon$ is expected squared law error, $p\geq1$,
$q\geq0$, $d=p+q$, and $L,M,\sigma>0$.
All norms in finite-dimensional spaces are Euclidean.

For the upper bound, set
\[
t=\frac{\sqrt{\varepsilon/2}}{LM},\qquad
\delta=\frac{\sqrt{\varepsilon/2}}{L\sqrt d},\qquad
m=\lceil\delta^{-1}\rceil,\qquad
n=\left\lceil\frac{4p\sigma^2L^2M^2}{\varepsilon^2}\right\rceil.
\]
Partition $D$ into $m^d$ cubes of side $1/m$, with centers
$z_i=(x_i,u_i)$. At each center, execute $n$ independent experiments
that reset to $x_i$, hold $u_i$ constant, and record the state once
at time $t$. Let $\overline Y_i$ be their mean and define
\[
\widehat F(z)=\operatorname{proj}_{\overline B(0,M)}
\left(\frac{\overline Y_i-x_i}{t}\right)
\quad\text{for }z\text{ in cube }i.
\]
The globally bounded Lipschitz field has a unique global solution.
Since $\|x(s)-x_i\|\leq Ms$,
\[
\left\|\frac{x(t)-x_i}{t}-F(x_i,u_i)\right\|
\leq\frac1t\int_0^tL\|x(s)-x_i\|\,ds
\leq\frac{LMt}{2}.
\]
For $z$ in cube $i$, the distance to its center is at most
$\delta\sqrt d/2$. Before projection the mean-square error is bounded by
\[
\frac{p\sigma^2}{nt^2}
+\left(\frac{LMt}{2}+\frac{L\delta\sqrt d}{2}\right)^2
\leq\frac\varepsilon2+
\left(\sqrt{\frac\varepsilon8}+\sqrt{\frac\varepsilon8}\right)^2
=\varepsilon.
\]
Projection onto the closed Euclidean ball does not increase distance to
$F(z)$, which belongs to that ball. Integration over $D$ therefore gives
the required uniform risk bound.

Define
\[
\varepsilon_{\rm up}
=\min\{2L^2d,\ 2L^2M^2,\ 2\sqrt p\,\sigma LM\}.
\]
For $\varepsilon\leq\varepsilon_{\rm up}$, we have $\delta\leq1$,
$t\leq1$, $m\leq2/\delta$, and
$n\leq8p\sigma^2L^2M^2/\varepsilon^2$. There are $m^dn$ experiments,
each of cost $1+c_{\rm r}+c_{\rm t}t\leq1+c_{\rm r}+c_{\rm t}$. Their total cost is at most
\[
8p(8d)^{d/2}(1+c_{\rm r}+c_{\rm t})
\sigma^2L^{d+2}M^2\varepsilon^{-(d+4)/2}.
\]
Thus one may take $c_2=8p(8d)^{d/2}$. In particular,
$\varepsilon_{\rm up}$ does not depend on the reset or elapsed-time
costs.

Canceling pulse alternatives also appear in the ODE lower bounds of
\citet{schoetz_siebel2026}. Here an antisymmetric perturbation of the
drift coordinate gives equal crossing times for its two signs even in
one state dimension. We require a bound uniform over arbitrary reset
locations and readout times, including resets inside a perturbation,
so it can be used conditionally along an adaptive transcript.

For the lower bound, write $z=(x,u)$ and designate the first state
coordinate by $z_1=x_1$. This coordinate exists because $p\geq1$.
Let $e_1\in\mathbb R^p$ be its unit vector, and set
\[
G=\pi\sqrt{d+3},\qquad P=2^{-d},\qquad v=M/2.
\]
Define functions on the real line by
\[
\psi_1(s)=\begin{cases}\sin(2\pi s),&0<s<1,\\0,&\text{otherwise},\end{cases}
\qquad
\psi_0(s)=\begin{cases}\sin(\pi s),&0<s<1,\\0,&\text{otherwise},\end{cases}
\]
and put $\psi(z)=\psi_1(z_1)\prod_{k=2}^d\psi_0(z_k)$, with an
empty product equal to one. This function is supported on the unit cube,
is $G$-Lipschitz, satisfies $|\psi|\leq1$, and obeys
$\int_{[0,1]^d}\psi^2=P$.

Choose an amplitude
$0<h\leq\min\{L/(4G),M/4\}$ and define
\[
m=\left\lfloor\frac{L}{2Gh}\right\rfloor,\qquad
\ell=1/m,\qquad N=m^d.
\]
Then $m\geq L/(4Gh)\geq1$ and $\ell\leq4Gh/L$.
Partition $D$ into these $N$ cubes, with lower corners $z_j$,
and let $\phi_j(z)=\psi((z-z_j)/\ell)$, extended by zero outside
its cube. For $\theta\in\{-1,1\}^N$, define
\begin{equation}
F_\theta(z)=\left(v+h\sum_{j=1}^N\theta_j\phi_j(z)\right)e_1.
\label{eq:budget-hypercube}
\end{equation}
All fields belong to $\mathcal F_{L,M}$: their speeds are between
$v-h$ and $v+h$, hence between $v/2$ and $3v/2$, and their norms
are at most $M$. Within a cube the Lipschitz constant is at most
$hG/\ell\leq L/2$. Because the perturbations vanish on all cube
boundaries, this bound extends across cubes and to the complement of
$D$. Their supports have disjoint interiors, so
\begin{equation}
\rho^2(F_\theta,F_{\theta'})
=4h^2\ell^dP\,\operatorname{Ham}(\theta,\theta'),
\label{eq:budget-hamming-distance}
\end{equation}
where $\operatorname{Ham}$ is Hamming distance.

Under \eqref{eq:budget-hypercube}, only $x_1$ moves.
The transverse state coordinates and the constant action fix one
column of cubes. Along this column, $x_1$ increases strictly.
For two worlds differing only in bit $j$, the trajectories coincide
until they enter cube $j$, unless the reset already lies inside it.
When they enter at its left face, their exact crossing times are equal.
Indeed, if $a\in[0,1]$ is the fixed product of the transverse sine
factors, those times are
\begin{equation}
T_\pm=\ell\int_0^1\frac{ds}{v\pm ha\sin(2\pi s)},
\qquad T_+=T_-.
\label{eq:budget-exact-crossing}
\end{equation}
The equality follows by substituting $s\mapsto1-s$.
It is exact, for every permissible $h$ and every transverse
coordinate, including $a=0$. After leaving the cell, the two
trajectories therefore coincide again at the same time and state;
they solve the same future scalar equation. A fully traversed
perturbation leaves no residual displacement or time delay.

During a crossing, both scalar speeds, evaluated at the respective
trajectory positions, lie in $[v-h,v+h]$ and therefore differ by at most $2h$.
Each crossing lasts at most $\ell/(v-h)\leq2\ell/v$, giving
\begin{equation}
|x_1^+(t)-x_1^-(t)|\leq4h\ell/v.
\label{eq:budget-displacement}
\end{equation}
Write $z_\perp=(x_2,\ldots,x_p,u_1,\ldots,u_q)$ for the fixed
transverse coordinates, with an empty tuple when $p=1,q=0$.
We also need this bound when the reset lies strictly inside a cell,
because a partial crossing need not cancel. If the cell's right
endpoint is $\beta$ and its reset coordinate is $s_0$, the difference
of the remaining crossing times has absolute value at most
\[
\left|\int_{s_0}^{\beta}
\left[\frac{1}{v+h\phi_j(s,z_\perp)}
-\frac{1}{v-h\phi_j(s,z_\perp)}\right]ds\right|
\leq\frac{2h\ell}{v^2-h^2}
\leq\frac{8h\ell}{3v^2}.
\]
Before both trajectories leave the cell,
\eqref{eq:budget-displacement} holds by the same speed and duration
bounds. After both leave, they follow the same autonomous scalar
flow with a time shift of at most $8h\ell/(3v^2)$. Its speed is at
most $3v/2$, so \eqref{eq:budget-displacement} continues to hold
at every later time.

Fix a world $\theta$ and any reset, constant action, and readout time.
Let $\theta^{(j)}$ denote the sign vector with bit $j$ flipped,
and let $\Delta_j$ be the difference between the corresponding
state-readout means. By \eqref{eq:budget-exact-crossing},
$\Delta_j$ can be nonzero for at most two bits: the cell containing
the reset, if it is only partially crossed, and the cell occupied
by the trajectory under $\theta$ at the readout time.
For a fully entered cell, the entrance and exit times are the same
in both worlds, which justifies this statement also when their
positions inside the cell differ. Only the first state coordinate
contributes to $\Delta_j$. Thus
\begin{equation}
\sum_{j=1}^N\|\Delta_j\|^2
\leq32h^2\ell^2/v^2,
\qquad
\sum_{j=1}^N\frac{\|\Delta_j\|^2}{2\sigma^2}
\leq K_h,
\quad
K_h=\frac{1024G^2h^4}{\sigma^2L^2M^2}.
\label{eq:budget-information-per-readout}
\end{equation}
The bound is uniform over duration, including readouts taken after
the trajectory has left $D$.

We now allow adaptive selection of experiments, within-experiment
readout times, and stopping. A constant action is still selected
for each experiment, and there is no observation other than the
paid readings. Let $P_\theta^\pi$ be the transcript distribution
for a policy $\pi$ with pathwise cost at most $B$.
At a fixed realized history the command kernels are the same in
all worlds. Conditional on that history and the selected command,
each readout is Gaussian with covariance $\sigma^2I_p$ and the
deterministic flow mean determined by the commands in its current
experiment. There are at most $\lfloor B\rfloor$ readouts.
Padding a stopped transcript by common dummy symbols, the
conditional divergence chain rule gives
\[
\operatorname{KL}(P_\theta^\pi\Vert P_{\theta^{(j)}}^\pi)
=\mathbb E_\theta\sum_r
\frac{\|\Delta_{j,r}\|^2}{2\sigma^2}.
\]
This identity uses natural logarithms; the divergences are in nats.
The differences in this expression compare two worlds at the same
realized history and command sequence. Consequently the pointwise
bound \eqref{eq:budget-information-per-readout} applies before
taking the expectation, and
\begin{equation}
\sum_{j=1}^N
\operatorname{KL}(P_\theta^\pi\Vert P_{\theta^{(j)}}^\pi)
\leq B K_h\qquad\text{for every }\theta.
\label{eq:budget-adaptive-sum}
\end{equation}
The hard budget is sufficient for this argument even when $c_{\rm r}=0$
or $c_{\rm t}=0$, because each readout still costs one unit.

For completeness, the testing reduction is given explicitly.
For any estimator $\widehat F$ with finite risk, choose a nearest
hypercube field $F_{\widehat\theta}$ in $\rho$, breaking ties by
a fixed order. The triangle inequality and
\eqref{eq:budget-hamming-distance} imply
\[
\rho^2(F_\theta,\widehat F)
\geq h^2\ell^dP\,
\operatorname{Ham}(\theta,\widehat\theta).
\]
An estimator that can output fields outside $L^2(D;\mathbb R^p)$
has infinite loss there and cannot improve the lower bound.
For each fixed $j$ and $\theta_{-j}$, the sum of the two testing
errors is at least
$1-\operatorname{TV}(P_{(\theta_{-j},+)}^\pi,
P_{(\theta_{-j},-)}^\pi)$.
Pinsker's inequality, followed by Cauchy--Schwarz and
\eqref{eq:budget-adaptive-sum}, yields
\begin{equation}
2^{-N}\sum_\theta
\mathbb E_\theta\operatorname{Ham}(\theta,\widehat\theta)
\geq\frac12\left(N-\sqrt{NBK_h}\right).
\label{eq:budget-assouad}
\end{equation}
To check the factor for adaptive designs, the sum over coordinates
of the average one-sided neighbor divergences is bounded as
\[
\sum_j2^{-(N-1)}\sum_{\theta_{-j}}
\operatorname{KL}(P_{(\theta_{-j},+)}^\pi\Vert
P_{(\theta_{-j},-)}^\pi)
\leq2\cdot2^{-N}\sum_\theta\sum_j
\operatorname{KL}(P_\theta^\pi\Vert P_{\theta^{(j)}}^\pi)
\leq2BK_h.
\]
This is all that \eqref{eq:budget-assouad} uses; individual
coordinates need not receive equal experimental effort.
If $B\leq N/(4K_h)$, the average Hamming risk is at least $N/4$.
Since $N\ell^d=1$, every such procedure satisfies
\begin{equation}
\sup_\theta\mathbb E_\theta\rho^2(F_\theta,\widehat F)
\geq Ph^2/4.
\label{eq:budget-final-risk}
\end{equation}

Finally, define
\[
\varepsilon_{\rm low}
=\frac{P}{8}\min\left\{\frac{L^2}{16G^2},\frac{M^2}{16}\right\},
\qquad
\varepsilon_0=\min\{\varepsilon_{\rm up},\varepsilon_{\rm low}\}.
\]
For $0<\varepsilon\leq\varepsilon_0$, set $h=\sqrt{8\varepsilon/P}$.
The construction is valid and \eqref{eq:budget-final-risk}
is at least $2\varepsilon$. A procedure with risk at most
$\varepsilon$ must therefore have $B>N/(4K_h)$, giving
\[
\begin{split}
B^\star_{\rm state}(\varepsilon)
&\geq\frac{N}{4K_h}
\geq
\frac{\sigma^2L^{d+2}M^2}{4096\,4^dG^{d+2}}
h^{-(d+4)}\\
&=c_1\sigma^2L^{d+2}M^2\varepsilon^{-(d+4)/2},\qquad
c_1=\frac{(P/8)^{(d+4)/2}}{4096\,4^dG^{d+2}}.
\end{split}
\]
This proves the theorem. The exact crossing symmetry is essential
to the proof as written: a generic zero-mean perturbation need
not have equal crossing times under its two signs.

\begin{proof}[Proof of Proposition~\ref{prop:velocity}]
For the upper bound, a unit readout returns $F(x_0,u)+\xi$ at cost $1+c_{\rm r}$. With $\delta$ as above, put $m=\lceil\delta^{-1}\rceil$ and $n=\lceil2p\sigma^2/\varepsilon\rceil$, and at each of the $m^d$ cube centers average $n$ readouts. The cellwise estimate has variance $p\sigma^2/n\le\varepsilon/2$ and interpolation error at most $L\delta\sqrt d/2=\sqrt{\varepsilon/8}$, so the risk is at most $\varepsilon/2+\varepsilon/8\le\varepsilon$. If $\varepsilon\le\min\{2L^2d,2p\sigma^2\}$, then $m\le2/\delta$ and $n\le4p\sigma^2/\varepsilon$, so the total cost is at most
\[
 4p(8d)^{d/2}(1+c_{\rm r})\sigma^2L^d\varepsilon^{-(d+2)/2};
\]
one may take $c_2'=4p(8d)^{d/2}$.

For the lower bound, replace $\psi$ by the nonnegative bump $\psi_0^{\otimes d}$ in~\eqref{eq:budget-hypercube}. It is $G'$-Lipschitz with $G'=\pi\sqrt d$, vanishes outside the open unit cube, is bounded by one, and its squared integral is again $P=2^{-d}$. With $G'$ in place of $G$, the fields, the Hamming identity~\eqref{eq:budget-hamming-distance}, and the testing reduction are unchanged. At each query location only one bump can be nonzero, and flipping its sign changes the mean by at most $2h$, so the sum of neighbor divergences per reading is at most $K'=2h^2/\sigma^2$. Each reading costs at least one unit, so the adaptive argument~\eqref{eq:budget-adaptive-sum}--\eqref{eq:budget-final-risk} applies with $K'$ in place of $K_h$. Let
\[
 \varepsilon_{\rm low}'=\frac P8\min\left\{\frac{L^2}{16G'^2},\frac{M^2}{16}\right\}
\]
and $h=\sqrt{8\varepsilon/P}$ for $0<\varepsilon\le\varepsilon'_{\rm low}$. A procedure with risk at most $\varepsilon$ must have
\[
 B>\frac{N}{4K'}=\frac{N\sigma^2}{8h^2}
 \ge\frac{(P/8)^{(d+2)/2}}{8\cdot4^dG'^d}\,\sigma^2L^d\varepsilon^{-(d+2)/2},
\]
using $N=m^d\ge(L/(4G'h))^d$. Thus $c_1'=(P/8)^{(d+2)/2}/(8\cdot4^dG'^d)$ and $\varepsilon_0'=\min\{2L^2d,2p\sigma^2,\varepsilon'_{\rm low}\}$.
\end{proof}
Both rate statements assume positive noise with fixed variance.

\section{Case-Study Methods and Reproducibility}
\label{app:experiments}

The source file \texttt{reproduce.py} generates the original orbit-selection, calibration, and scaling results, using seed 20260912. Machine-readable results are saved in \texttt{results.json}; \texttt{scaling\_raw.npz} and \texttt{calibration\_raw.npz} retain the sampled law errors. The execution used NumPy 2.3.5, SciPy 1.17.0, and Matplotlib 3.10.8. The additional release-design and continuous-force studies use \texttt{extended\_orbit.py}; their full methods and records are described at the end of this appendix. Together the two scripts reproduce all reported numerical studies.

For the circular example, normalize the reference radius $r_0$ and acceleration $g_0$ to one. The common orbit is $q(t)=(\cos t,\sin t)$ with velocity $\dot q(t)=(-\sin t,\cos t)$. Both fields give acceleration $-q(t)$ on this orbit. Away from the origin, both define ordinary smooth dynamics; the displayed release trajectories remain within the regular annulus $0.7<\|q\|<2.1$. The harmonic release radius is $r_H(t)=2\cos t$. The inverse-square release has the parametric form
\[
r_G(t)=2\cos^2\theta,\qquad
t=2(\theta+\sin\theta\cos\theta),\quad 0\leq\theta<\pi/2.
\]
The implementation solves the strictly increasing scalar time relation by Brent's method. The noise experiment observes radius once at $t=0.4$ after each independent release. The equal-prior likelihood-ratio test compares the average readout to the midpoint of the two noiseless means. Its two conditional error probabilities coincide. We simulate 200,000 independent averages under one candidate at each $n\in\{1,2,4,8,16,32,64,128\}$, and evaluate the analytic normal-tail probability separately. A zero observed error at the largest budgets is a finite Monte Carlo outcome, not a claim of zero true risk. The two first-order fields have the same velocity component, so their difference is entirely in acceleration. For uniform radial weighting on $[1,2]$, the squared acceleration-law distance is $\int_1^2(r-r^{-2})^2\,dr=21/8-2\log2$. A decoder restricted to selecting a candidate has squared-law risk equal to this distance times its testing error. An unrestricted squared-loss decoder can perform better by interpolating; its exact uninformative two-point minimax risk is one quarter of the squared distance.

For the gain example, write $\underline c=0.75$ and $\overline c=1.25$. The unknown world is indexed by $c\in[\underline c,\overline c]$, with $g=1/c$ and known drift coefficient one. Dynamics start from $x=0$ and satisfy $\dot x=-x+gu$; the observation map is $H_c(x)=cx$, with known zero offset. Thus $z=cx$ satisfies $\dot z=-z+u$ for every $c$. Any sequence of dynamic observations, including adaptively chosen inputs, has a distribution independent of $c$. With a uniform prior, the optimal uncalibrated squared-loss decoder is
\[
\hat g_0=\mathbb E[c^{-1}]
=\frac{\log(\overline c/\underline c)}{\overline c-\underline c},
\quad
R_0=\frac13\left[
\frac{\underline c^{-1}-\overline c^{-1}}{\overline c-\underline c}
-\left(\frac{\log(\overline c/\underline c)}{\overline c-\underline c}\right)^2
\right].
\]
The loss factor $1/3$ is $\mathbb E[u^2]$ under the declared uniform evaluation distribution; the common drift cancels. This baseline uses the best possible prior decoder and exact observable dynamics, so its error cannot be attributed to an optimizer or an arbitrary sensor-gain convention.

Each calibration episode clamps the state at the known physical value one and observes $c+\xi$, where $\xi\sim N(0,0.08^2)$ independently across episodes. We charge one unit for the clamp and one for the readout. At $n\in\{1,4,16,64,256,1024\}$, the sample mean is generated directly from $N(c,0.08^2/n)$, clipped to the known interval, and inverted to obtain $\hat g$. We draw 1,024 independent gains once and generate 128 independent noise realizations per gain and budget. All squared errors are evaluated analytically as $(\hat g-g)^2/3$. Displayed standard errors use the 1,024 world-level averages, preserving the grouping of noise replicates within a world. The last-three-budget slope uses ordinary least squares on logarithms of mean loss and budget. At small noise the integrated leading term is
\[
\mathbb E\,\loss(\lawhat,F)^2
\sim\frac{\sigma^2}{3n}\mathbb E[c^{-4}]
=\frac{\sigma^2(\underline c^{-3}-\overline c^{-3})}
{9n(\overline c-\underline c)}.
\]
This is a prior-averaged statement for the linear gain family, not a numerical worst-case guarantee for the entire nonlinear family.

For the scaling experiment, define the continuous triangle wave
\[
\operatorname{tri}(s)=1-2|\,(s\bmod2)-1\,|.
\]
The two globally defined fields are
\[
f_1(x)=0.65+0.12\operatorname{tri}(2x+0.1),\qquad
f_2(x,u)=f_1(x)+0.08\cos(2\pi u).
\]
They are evaluated on $[0,1]$ and $[0,1]^2$, respectively, with the uniform probability measure. The first field has bounds $|f_1|\leq0.77$ and Lipschitz constant at most $0.48$. The second has bounds $|f_2|\leq0.85$ and Lipschitz constant at most $\sqrt{0.48^2+(0.16\pi)^2}$. These are analytical bounds, not grid estimates. In both cases the scalar flow remains positive, and for each fixed input the field is affine in $x$ between known triangle-wave kinks.

The grid resolutions are $m\in\{8,12,16,24,32,48,64,96,128\}$. For each of the $m^d$ cells, the state experiment resets to its state midpoint and holds its input midpoint constant when an input is present. It uses $t=1/m$ and $n=\lceil\sigma^2m^4\rceil$ noisy endpoint readouts, with $\sigma=0.05$. The cell estimate is $(\bar y-x_0)/t$. The velocity comparator directly averages $n=\lceil\sigma^2m^2\rceil$ observations of $f(x_0,u_0)+\xi$. The fixed-duration comparator uses $t=0.25$ and $n=\lceil\sigma^2m^2/t^2\rceil$. All three estimators are piecewise constant on the same cells; no function fitting, architecture choice, clipping, or hyperparameter tuning is performed. Every independent episode incurs one unit for reset, one for readout, and its elapsed time, so
\[
B_{\rm state}=(2+t)m^d n,\qquad
B_{\rm velocity}=2m^d n.
\]
Repeating an experiment $n$ times produces an exactly Gaussian sample mean with variance $\sigma^2/n$; the implementation samples this sufficient statistic directly. Consequently, large plotted acquisition budgets represent the number and costs of observations in the statistical experiment, not a claim that billions of individual measurements were physically executed or stored.

Within an affine segment with slope $k\in\{-0.48,0.48\}$ and initial speed $v$, displacement after duration $t$ is $v(e^{kt}-1)/k$. The time to the next kink at distance $\ell$ is $\log(1+k\ell/v)/k$. Chaining these exact expressions computes the endpoint without time discretization. Six checks spanning both fields, short intervals, and kink crossings agree with high-accuracy DOP853 integrations to less than $2.70\times10^{-12}$. For a cellwise constant estimate $\hat f_j$, the mean squared error within cell $j$ is $(\hat f_j-\mu_j)^2+v_j$, where $\mu_j$ and $v_j$ are the exact within-cell mean and variance of the true field. Each cell has volume $m^{-d}$, so the integrated law error over the unit domain is
\[
\loss(f,\hat f)^2=m^{-d}\sum_{j=1}^{m^d}\bigl[(\hat f_j-\mu_j)^2+v_j\bigr].
\] Triangle-wave moments are integrated by splitting at their kinks; cosine moments use their elementary antiderivatives. Thus squared law errors have no numerical evaluation-grid approximation. We sample 64 independent noises per configuration and report their means and standard errors. Slopes use only the five finest grids, $m\geq32$. The schedules are designed to make sampling noise and approximation error decrease with cell width; agreement of their empirical slopes with the construction is an implementation check and a visual illustration, not independent verification of a universal minimax exponent.

The additional orbital studies use \texttt{extended\_orbit.py}, with master seed 20260915. They produce \texttt{extended\_results.json}, two raw NPZ archives, and the release-design and codebook figures in Section~\ref{sec:cases}. The software versions are the same as those listed above. Their purpose is to vary specific elements of the apparatus or decoder while retaining a physical ground truth that can be evaluated independently.

For a release from rest at radius $R$, write $m_H(R,t)=R\cos t$ for the harmonic response. Before collision with the origin, the inverse-square response has the parameterization
\begin{equation}\label{eq:general-radial-release}
 m_G(R,t)=R\cos^2\psi,\qquad
 t=\sqrt{R^3/2}\,(\psi+\sin\psi\cos\psi),
 \qquad 0\le\psi<\pi/2.
\end{equation}
The auxiliary angle $\psi$ is solved by Brent's method. The collision time is $\frac{\pi}{2}\sqrt{R^3/2}$, greater than every admitted readout time in this study. The design grid contains 101 equally spaced radii from one to two and 111 times from $0.05$ to $0.6$. Because the two response distributions are scalar Gaussians with common variance, their optimal equal-prior testing error follows directly from the midpoint threshold. Equations~\eqref{eq:release-design-score}--\eqref{eq:release-onepercent} then give the score and integer budget. A zero mean gap leaves error $1/2$ at every budget and has no finite $1\%$ threshold. The displayed finite grid contains no exactly zero gap; the equal-reading curve is located separately from the signed responses. Independent DOP853 integrations at 12 radius--time pairs agree with~\eqref{eq:general-radial-release} to within $1.09\times10^{-13}$.

The unit-radius release deserves care. Its initial velocity is zero, whereas the indistinguishable circular preparation has unit tangential speed. At $R=1$, Taylor expansion gives $m_H(1,t)=1-t^2/2+t^4/24+O(t^6)$ and $m_G(1,t)=1-t^2/2-t^4/12+O(t^6)$. Their first acceleration terms agree, but their radius gap is $-t^4/8+O(t^6)$. The small fourth-order signal explains why short readouts from this preparation require many repetitions. Some of these release trajectories pass below radius one. Their positions evolve in the punctured plane, with velocity included in the full state. The field error is evaluated on the compact domain specified in Section~\ref{sec:problem}, whose positions lie in an annulus; the experimental trajectories need not stay within that evaluation region.

For the one-readout ambiguity at time $0.6$, define $g(R)=m_G(R,0.6)-m_H(R,0.6)$. It is continuous on $[1,1.1]$, with $g(1)\simeq-0.01786462$ and $g(1.1)\simeq0.03589940$. The intermediate value theorem therefore establishes an exact zero; a numerical root gives $R_*\simeq1.03074836417$ and shared terminal radius $0.85071333496$. The small numerical residual is a computation check, rather than the existence argument. At the same release radius, the signed response gap at $0.3$ is approximately $0.00308324978$. With independent non-invasive readings at $0.3$ and $0.6$, let $\boldsymbol d$ be the two-dimensional difference of means. Averaging $n$ independent episodes gives equal-prior error $\Phi(-\sqrt n\|\boldsymbol d\|_2/(2\sigma))$. Applying the same normal quantile as in~\eqref{eq:release-onepercent}, with $\sigma=0.1$, gives 22,772 episodes. Each costs $1+2+0.6=3.6$ units and contains two readings; thus 45,544 scalar observations cost $81,979.2$ units. The additional observation restores the missing distinction but does not make it inexpensive.

For the continuous family~\eqref{eq:continuous-force-family}, radial motion from rest obeys $\ddot r=-(1-\theta)r-\theta/r^2$. The response $m(\theta)$ at $R=2$, $t=0.6$ is computed with DOP853, relative tolerance $2\times10^{-12}$ and absolute tolerance $2\times10^{-14}$. Endpoint checks at $\theta=0,1$ agree with the harmonic and inverse-square expressions to within $8.4\times10^{-14}$. A 1,025-point parameter grid gives a strictly increasing sampled response curve. This is a numerical diagnostic over the tested grid, not a proof of injectivity for every real parameter. The experiment uses stored responses rather than assuming an analytic inverse of this curve.

For each codebook size $K_c$, use parameter values $\theta_j=(j-1)/(K_c-1)$, $j=1,\ldots,K_c$, with means $m_j=m(\theta_j)$. The sampled means are increasing. Define response thresholds $b_0=-\infty$, $b_{K_c}=+\infty$, and $b_j=(m_j+m_{j+1})/2$ for $1\le j<K_c$. At a true parameter $\theta$, the sample mean has law $N(m(\theta),\sigma^2/n)$. Its probability of being decoded to codeword $j$ is
\begin{equation}\label{eq:codebook-gaussian-cells}
 p_j(\theta;n,\sigma)
 =\Phi\!\left(\frac{\sqrt n[b_j-m(\theta)]}{\sigma}\right)
 -\Phi\!\left(\frac{\sqrt n[b_{j-1}-m(\theta)]}{\sigma}\right).
\end{equation}
Consequently the noise-averaged law error of this finite decoder is
\[
 \delta_{\rm law}^2\sum_{j=1}^{K_c}p_j(\theta;n,\sigma)(\theta_j-\theta)^2.
\]
The Gaussian probabilities and loss calculation have no Monte Carlo approximation; their response means retain the stated numerical ODE tolerance. The noiseless decoder instead picks the stored response closest to $m(\theta)$. Its resulting error is reported separately. It need not pick the parameter-nearest grid point, because the response is not linear in $\theta$.

The 32 fixed true parameters comprise $0.37$, $0.63$, and 30 independent uniform draws generated with seed 20260955; all are checked to lie off each codebook. We use $n\in\{1,4,16,64,256,1024,4096,16384\}$ and 50,000 independent Gaussian draws for every true parameter and design. At noise $\sigma=0.1$ all three codebooks are evaluated; the 65-entry codebook is also evaluated at $\sigma=0.03,0.3$. Within each $n$, the same standardized noise draws are reused across codebooks and noise levels. Draws are independent across true parameters and across $n$. A mean over these 32 fixed truths is an empirical design average; it is not an integral over the uniform parameter distribution or a supremum over the continuous family.

Let $V_i$ be the sample variance of individual squared law errors at true parameter $i$ in one design. The reported Monte Carlo standard error of the average is $\sqrt{\sum_{i=1}^{32}V_i/50{,}000}/32$. It measures uncertainty from noise simulation conditional on the fixed truths. The raw archive stores decoder-index histograms for every truth and design, from which all realized loss moments can be reconstructed, together with the means, codebooks, and seeds. Aggregate Monte Carlo estimates and probability-based calculations differ by at most 3.48 estimated standard errors over the 40 compared designs. A missing rare mistake in an individual finite simulation does not imply zero true error. The noiseless codebook levels describe finite approximation by these particular decoders; an estimator that interpolates between entries need not share them.

The study uses simulated systems, analytic ground-truth fields, and synthetic measurement noise. It involves no human participants or personal data. The evidence supports the stated mathematical mechanisms and the implemented estimators; it does not certify a learned controller for physical deployment.

\end{document}